\documentclass[runningheads,UKenglish]{llncs}
\pdfoutput=1
\usepackage{graphicx}
\usepackage{booktabs}
\usepackage{longtable}
\usepackage{adjustbox}
\usepackage{pdflscape}
\usepackage{rotating}

\usepackage{amsmath} 
\usepackage{amssymb} 

\usepackage{stmaryrd} 
\usepackage{cancel} 

\usepackage{stackengine}

\usepackage{mathrsfs} 

\usepackage[ruled,linesnumbered]{algorithm2e}
\SetEndCharOfAlgoLine{}

\SetCommentSty{mycommfont}
\SetKwRepeat{Do}{do}{while}
\SetKw{KwBy}{by}

\makeatletter
\newcommand*{\algotitle}[2]{%
  \stepcounter{algocf}%
  \hypertarget{algocf.title.\theHalgocf}{}%
  \NR@gettitle{#1}%
  \label{#2}%
  \addtocounter{algocf}{-1}%
}
\makeatother

\usepackage[activate={true,nocompatibility},final,tracking=true,kerning=true,spacing=true,factor=1100,stretch=10,shrink=10]{microtype}

\usepackage[hidelinks]{hyperref} 
\usepackage{cleveref} 

\usepackage{tikz}
\usetikzlibrary{automata,positioning,arrows}
\usetikzlibrary{external}
\usetikzlibrary{cd} 
\usetikzlibrary{arrows.meta}

\usepackage{etoolbox} 

\newtoggle{comments}

\togglefalse{comments} 

\newcommand{\set}[1]{\left\{ #1\right \}}
\newcommand{\oracle}{\mathcal{O}}

\newcommand{\dataDomain}{\mathcal{D}}
\newcommand{\D}{\mathcal{D}}

\newcommand{\relationSet}{\mathcal{R}}
\newcommand{\R}{\mathcal{R}}

\newcommand{\nat}{\mathbb{N}}

\newcommand{\satisfies}{\vDash}

\newcommand{\tree}{\mathcal{T}}
\newcommand{\lang}{\mathcal{L}}

\newcommand{\tuple}[1]{\langle{} #1 \rangle}
\newcommand{\dbracket}[1]{\llbracket{} #1 \rrbracket}
\newcommand{\parens}[1]{\left(#1 \right)}
\newcommand{\iso}{\simeq}
\newcommand{\fin}{\hfill$\lrcorner$}

\newcommand{\M}{\mathscr{M}}
\newcommand{\X}{\mathcal{X}}
\newcommand{\Structure}{\mathcal{S}}
\newcommand{\V}{\mathcal{V}}

\newcommand{\constraints}{\mathit{constraints}}

\newcommand{\Acts}{\mathit{Acts}}
\newcommand{\Vals}{\mathit{Vals}}

\newcommand{\push}{{\sf Push}}
\newcommand{\pop}{{\sf Pop}}

\usepackage{color}
\usepackage{centernot}

\usepackage{lmodern} 

\definecolor{ochre}{rgb}{0.8, 0.47, 0.13}

\usepackage[square,comma,numbers,sort&compress]{natbib}

\usepackage[%
    font={small,sf},
    labelfont=bf,
    format=hang,    
    format=plain,
    margin=0pt,
    width=0.8\textwidth,
]{caption}
\usepackage[list=true]{subcaption} 

\usepackage{listings}
\lstdefinestyle{py} {
    basicstyle  = \ttfamily\scriptsize,
    breaklines  = true,
    frame       = single,
    numbers     = left,
    tabsize     = 4,
}

\usepackage[title]{appendix}

\newcommand{\inv}{^{\textnormal{-}1}}
\newcommand{\etal}{\xspace{et al.}}

\newif\iflong
\longtrue

\title{Grey-Box Learning of Register Automata}

\author{Bharat Garhewal\inst{1}\thanks{Supported by NWO TOP project 612.001.852
    ``Grey-box learning of Interfaces for Refactoring Legacy Software (GIRLS)''.}
  \and Frits Vaandrager\inst{1} \and Falk Howar\inst{2} \and \\ Timo Schrijvers\inst{1} \and Toon Lenaerts\inst{1} \and Rob Smits\inst{1}}

\institute{Radboud University, Nijmegen, The Netherlands\\
\email{\{bharat.garhewal, frits.vaandrager\}@ru.nl}
\and
Dortmund University of Technology}

\authorrunning{B. Garhewal et al.}

\begin{document}


%
\maketitle
\begin{abstract}
	\iflong
	Model learning (a.k.a.\ active automata learning) is a highly effective technique for obtaining black-box finite state  models of software components. Thus far, generalization to infinite state systems with inputs and outputs that carry data parameters has been challenging. Existing model learning tools for infinite state systems face scalability problems and can only be applied to restricted classes of systems (register automata with equality/inequality).
	In this article, we 
	\fi
	show how one can boost the performance of model learning techniques by extracting the constraints on input and output parameters from a run, and making this grey-box information available to the learner.
	More specifically, we provide new implementations of the tree oracle and equivalence oracle from the RALib tool, which
	use the derived constraints. We extract the constraints from runs of Python programs using an existing tainting library for Python, and compare our grey-box version of RALib with the existing black-box version on several benchmarks, including some data structures from Python's standard library.
	Our proof-of-principle implementation results in almost two orders of
  magnitude improvement in terms of numbers of inputs sent to the software
  system. Our approach, which can be generalized to richer model classes, also enables RALib to learn models that are out of reach of black-box techniques, such as combination locks.
\end{abstract}

%

\section{Introduction}\label{sec:introduction}

Model learning, also known as active automata learning, is a black-box
technique for constructing state machine models of software and hardware
components from information obtained through testing (i.e., 
providing inputs and observing the resulting outputs).  
Model learning has been successfully used in
numerous applications, for instance for generating conformance test
suites of software components \cite{HMNSBI2001}, finding mistakes in
implementations of security-critical protocols
\cite{FJV16,FiterauEtAl17,FH17}, learning interfaces of classes in
software libraries \cite{HowarISBJ12}, and checking that a legacy
component and a refactored implementation have the same behaviour
\cite{SHV16}. We refer to \cite{Vaa17,HowarS2018} for surveys and
further references.

%

In many applications it is crucial for models to 
describe \emph{control flow}, i.e., states of a component, 
\emph{data flow}, i.e., constraints on data parameters that
are passed when the component interacts with its environment, as well
as the mutual influence between control flow and data flow. 
Such
models often take the form of \emph{extended finite state machines}
(EFSMs). Recently, various techniques have been employed to extend
automata learning to a specific class of EFSMs called \emph{register
automata}, which combine control flow with guards and assignments to
data variables~\cite{Cassel2016,AJUV15,CEGAR12}. 

While these works demonstrate that it is theoretically 
possible to infer such richer models, the presented approaches
do not scale well and are not
yet satisfactorily developed for richer classes of models
(c.f. \cite{HowarJV19}):
Existing techniques
either rely on manually constructed mappers that abstract the data
aspects of input and output symbols into a finite alphabet, or
otherwise infer guards and assignments from black-box observations of
test outputs. The latter can be costly, especially for models where
control flow depends on test on data parameters in input: in this
case, learning an exact guard that separates two control flow branches
may require a large number of queries. 

One promising strategy for addressing the challenge of identifying 
data-flow constraints is to augment learning algorithms with 
white-box information extraction methods, which are able to 
obtain information about the System Under Test (SUT) at lower cost than black-box techniques.
Several researchers have explored this idea.
Giannakopoulou \etal\ \cite{acm2414956} develop an active learning algorithm 
that infers safe interfaces of software components with 
guarded actions. In their model, the teacher is implemented
using concolic execution for the identification of guards.
Cho \etal\  \cite{acm2028077} present MACE an approach for concolic exploration 
of protocol behaviour. The approach uses active automata learning
for discovering so-called deep states in the protocol behaviour.
From these states, concolic execution is employed in order to
discover vulnerabilities. 
Similarly, \citet{BotincanB13} present a learning algorithm for 
inferring models of stream transducers that integrates active 
automata learning with symbolic execution and counterexample-guided 
abstraction refinement. They show how the models 
can be used to verify properties of input sanitizers in 
Web applications. 
Finally,~\citet{acm2483783} extend the work of~\cite{acm2414956} 
and integrate knowledge obtained through static code analysis 
about the potential effects of component 
method invocations on a component's state
to improve the performance during symbolic
queries. 
So far, however, white-box techniques have never been integrated 
with learning algorithms for register automata.

In this article, we present the first active learning algorithm for a general class of register automata that uses white-box techniques.
More specifically, we show how dynamic taint analysis can be used to 
efficiently extract constraints on input and output parameters 
from a test, and how these constraints can be used to improve the performance of the $\mathit{SL}^{\ast}$ algorithm of~\citet{Cassel2016}.
The $\mathit{SL}^{\ast}$ algorithm generalizes the classical
\(\mathit{L}^*\) algorithm of \citet{DanaAngluinBasePaper} and has
been used successfully to learn register automaton models, 
for instance of Linux and Windows implementations of TCP~\cite{FH17}.
We have implemented the presented method on top of 
RALib~\cite{raLibFirstPaper}, a library that provides an 
implementation of the $\mathit{SL}^{\ast}$ algorithm.

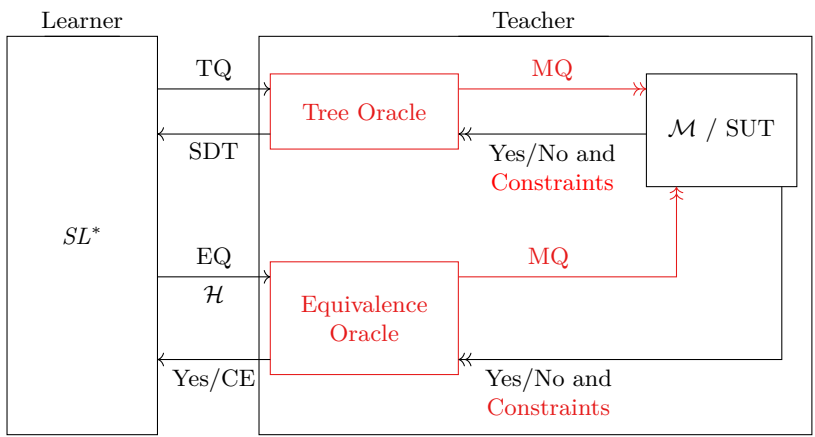
\begin{figure}[t]
  \centering
  \begin{tikzpicture}


\draw (1,-0.3) rectangle node[align=center] {$\mathit{SL}^{\ast}$} (3,5);
\draw[-] (1.5,5) -- node[above] {Learner} (2.5, 5);


\draw (4.35,-0.3) rectangle  (11.7,5);
\draw[-] (7,5) -- node[above] {Teacher} (9, 5);
\draw (9.5, 3) rectangle node[] {$\mathcal{M}$ / SUT} (11.5,4.5);


\draw[color=red!80!gray] (4.5,0.5) rectangle node[align=center] {Equivalence\\Oracle} (7,2);
\draw[color=red!80!gray] (4.5,3.5) rectangle node[] {Tree Oracle} (7,4.5);


\draw[->] (3,4.3) -- node[above] {TQ} (4.5, 4.3);
\draw[<-] (3,3.7) -- node[below] {SDT} (4.5, 3.7);

\draw[->] (3,1.8) -- node[above] {EQ} node[below] {$\mathcal{H}$} (4.5, 1.8);
\draw[<-] (3,0.7) -- node[below, align=center] {Yes/CE} (4.5, 0.7);


\draw[->>, color=red!80!gray] (7,4.3) -- node[above, color=red!80!gray] {MQ} (9.5, 4.3);
\draw[<<-, color=black] (7, 3.7) -- node[below, align=center] {Yes/No and\\\textcolor{red}{Constraints}} (9.5, 3.7);

\draw[->>, color=red!80!gray] (7, 1.8) node[xshift=1.2cm, above] {MQ} -| (9.9, 3);
\draw[<<-, color=black] (7, 0.7) node[xshift=1.2cm, below, align=center] {Yes/No and\\\textcolor{red!80!gray}{Constraints}} -| (11.3, 3);

\end{tikzpicture}

  \caption{MAT Framework (Our addition --- tainting --- in red): 
    Double arrows indicate possible multiple instances of a query made by an
    oracle for a single query by the learner. 
   \label{fig:matFramework}}
\end{figure}

The integration of the two techniques (dynamic taint analysis
and learning of register automata models) can be explained 
most easily with reference to the architecture of RALib, 
shown in \Cref{fig:matFramework}, which is a variation of the
\emph{Minimally Adequate Teacher} (MAT) framework of~\cite{DanaAngluinBasePaper}:
In the MAT framework, learning is viewed as a game in which a \emph{learner} has to infer the behaviour of an unknown register automaton $\mathcal{M}$ by asking queries to a \emph{teacher}.
We postulate $\mathcal{M}$ models the behaviour of a \emph{System Under Test (SUT)}.
In the learning phase, the learner (that is, $\mathit{SL}^{\ast}$) is allowed to ask questions to the
teacher in the form of \emph{tree queries} (TQs) and the teacher responds with
\emph{symbolic decision trees} (SDTs).
In order to construct these SDTs, the teacher uses a \emph{tree oracle}, which queries
the SUT with \emph{membership queries} (MQs) and receives a yes/no reply to each.
Typically, the tree oracle asks multiple MQs to answer a single tree query
in order to infer causal impact and flow of data values. 
Based on the answers on a number of tree queries, the learner constructs a \emph{hypothesis} in the form of a register
automaton \(\mathcal{H}\). The learner submits
\(\mathcal{H}\) as an \emph{equivalence query (EQ)} to the teacher, asking whether \(\mathcal{H}\) is equivalent to the SUT model $\mathcal{M}$.
The teacher uses an \emph{equivalence oracle} to answer equivalence queries.
Typically, the equivalence oracle asks multiple MQs to answer a single equivalence query.
If, for all membership queries, the output produced by the SUT is consistent with hypothesis $\mathcal{H}$,
the answer to the equivalence query is `Yes' (indicating learning is complete).
Otherwise, the answer `No' is provided, together with a \emph{counterexample} (CE) that indicates
a difference between  \(\mathcal{H}\) and $\mathcal{M}$. Based on this CE, learning continues.
In this extended MAT framework, we have constructed new implementations of 
the tree oracle and equivalence oracle that leverage the constraints on 
input and output parameters that are imposed by a program run:
dynamic tainting is used to extract the constraints on parameters 
that are encountered during a run of a program.
Our implementation learns models of Python programs,
using an existing tainting library for Python~\cite{abs-1810-08289}.
Effectively, the combination of the $\mathit{SL}^{\ast}$ with our new tree and equivalence oracles constitutes a \emph{grey-box} learning algorithm, since we only give the learner partial information about the internal structure of the SUT\@.

We compare our grey-box tree and equivalence oracles with the existing
black-box versions of these oracles on several benchmarks, including Python's
{\tt queue} and {\tt set} modules.
Our proof-of-concept implementation\footnote{Available at \url{https://bitbucket.org/toonlenaerts/taintralib/src/basic}.} results in almost two orders of magnitude
improvement in terms of numbers of inputs sent to the software system.
Our approach, which generalises to richer model classes,
also enables RALib to learn models that are completely out of reach 
for black-box techniques, such as combination locks.

\medskip
\noindent
\textbf{Outline:}~\Cref{sec:preliminaries} contains preliminaries;~\Cref{sec:tainting}
discusses tainting in our Python SUTs;~\Cref{sec:learningWithTainting} contains
the algorithms we use to answer TQs using tainting and the definition for the
tainted equivalence oracle needed to learn combination lock
automata;~\Cref{sec:experimentalEval} contains the experimental evaluation of
our technique; and~\Cref{sec:conclusion} concludes.

\iflong

\else

\fi


%
%
\section{Preliminary definitions and constructions}\label{sec:preliminaries}

This section contains the definitions and constructions necessary to understand active automata learning for models with dataflow.
We first define the concept of a \emph{structure}, followed by \emph{guards}, \emph{data languages}, \emph{register automata}, and finally \emph{symbolic decision trees}.

\begin{definition}[Structure]\label{def:theory}
    A structure \( \Structure = \tuple{R, \dataDomain,\relationSet}\) is a triple where $R$ is a set of relation symbols, each equipped with an arity,  \(\dataDomain \) is an infinite domain of data values, and \(\relationSet \) contains a distinguished $n$-ary relation $r^{\relationSet} \subseteq \dataDomain^n$ for each $n$-ary relation symbol $r \in R$.
\end{definition}

%
%

In the remainder of this article, we fix a structure \( \Structure = \tuple{R, \dataDomain,\relationSet}\), where $R$ contains a binary relation symbol $=$ and unary relation symbols $=c$, for each $c$ contained in a finite set $C$ of constant symbols, $\dataDomain$ equals the set $\nat$ of natural numbers, $=^{\relationSet}$ is interpreted as the equality predicate on $\nat$, and to each symbol $c \in C$ a natural number $n_c$ is associated such that  $(=c)^{\relationSet} = \{ n_c \}$.

Guards are a restricted type of Boolean formulas that may contain relation symbols from $R$.

\begin{definition}[Guards]
We postulate a countably infinite set $\V = \{ v_1, v_2,\ldots \}$ of \emph{variables}.
In addition, there is a variable $p \not\in\V$ that will play a special role as formal parameter of input symbols; we write $\V^+ = \V \cup \{ p \}$.
A {\em guard} is a conjunction of relation symbols and negated relation symbols over variables.
Formally, the set of \emph{guards} is inductively defined as follows:
\begin{itemize}
	\item 
	If $r\in R$ is an $n$-ary relation symbol and $x_1 ,\ldots, x_n$ are variables from $\V^+$, then $r(x_1,\ldots,x_n)$ and $\neg r(x_1,\ldots,x_n)$ are guards.
	\item
	If $g_1$ and $g_2$ are guards then $g_1 \wedge g_2$ is a guard.
\end{itemize}
Let $X \subset \V^+$.
We say that $g$ is a guard \emph{over} $X$ if all variables that occur in $g$ are contained in $X$.
A \emph{variable renaming} is a function $\sigma:X \rightarrow \V^+$.
If $g$ is a guard over $X$ then $g[\sigma]$ is the guard obtained by replacing each variable $x$ in $g$ by $\sigma(x)$.
\end{definition}

Next, we define the notion of a \emph{data language}. For this, we fix a finite set of \emph{actions} \(\Sigma \).
%
A \emph{data symbol} \(\alpha(d)\) is a pair consisting of an action \(\alpha \in \Sigma \) and a data value \(d \in \D \).
While relations may have arbitrary arity, we will assume that all actions have an arity of one to ease notation and simplify the text.
A \emph{data word} is a finite sequence of data symbols, and a \emph{data language} is a set of data words.
We denote concatenation of data words $w$ and $w'$ by \(w \cdot w'\), where \(w\) is the \emph{prefix} and \(w'\) is the \emph{suffix}.
\(\Acts(w)\) denotes the sequence of actions \(\alpha_1 \alpha_2 \ldots \alpha_n\) in \(w\), and \(\Vals(w)\) denotes the sequence of data values \(d_1 d_2 \ldots d_n\) in \(w\). 
We refer to a sequence of actions in $\Sigma^{\ast}$ as a \emph{symbolic suffix}. If $w$ is a symbolic suffix then we write $\dbracket{w}$ for the set of data words $u$ with $\Acts(u) = w$.

Data languages may be represented by \emph{register automaton}, defined below.

\begin{definition}[Register Automaton]\label{def:registerAutomaton}
    A Register Automaton (RA) is a tuple \(\M = (L, l_0, \X, \Gamma, \lambda)\) where
    \begin{itemize}
        \item \(L\) is a finite set of locations, with \(l_0\) as the initial location;

        \item \(\X\) maps each location \(l \in L\) to a finite set of registers \(\X(l)\);

        \item \(\Gamma \) is a finite set of transitions, each of the form \(\tuple{l, \alpha(p), g, \pi, l'}\), where
              \begin{itemize}
                  \item \(l, l'\) are source and target locations respectively,
                  \item \(\alpha(p)\) is a parametrised action,
                  \item \(g\) is a \emph{guard} over \(\X(l) \cup \set{p}  \), and
                  \item \(\pi \) is an assignment mapping from \(\X(l')\) to \(\X(l) \cup \set{p}\); and
              \end{itemize}

        \item \(\lambda \) maps each location in \(L\) to either accepting $(+)$ or rejecting $(-)$.
    \end{itemize}
We require that $\M$ is {\em deterministic} in the sense
that for each  location $l \in L$ and input symbol $\alpha \in \Sigma$,
the conjunction of the guards of any pair of distinct $\alpha$-transitions with source $l$ is not satisfiable.
$\M$ is \emph{completely specified} if for all \(\alpha \)-transitions out of a location, the disjunction of the guards of the \(\alpha \)-transitions is a tautology.
$\M$ is said to be \emph{simple} if there are no registers in the initial location, i.e., \(\X(l_{0}) = \varnothing \).
In this text, all RAs are assumed to be completely specified and simple, unless explicitly stated otherwise.
Locations $l \in L$ with $\lambda(l) = +$ are called \emph{accepting}, and locations with $\lambda(l) = -$ \emph{rejecting}.
\end{definition}

\begin{example}[FIFO-buffer]
	The register automaton displayed in~\Cref{fig:fifo} models a FIFO-buffer with capacity 2. It has three accepting locations $l_0$, $l_1$ and $l_2$ (denoted by a double circle), and one rejecting ``sink'' location $l_3$ (denoted by a single circle). Function $\mathcal{X}$ assigns the empty set of registers to locations $l_0$ and $l_3$, singleton set $\set{x}$ to location $l_1$, and set $\set{x, y}$ to $l_2$.
	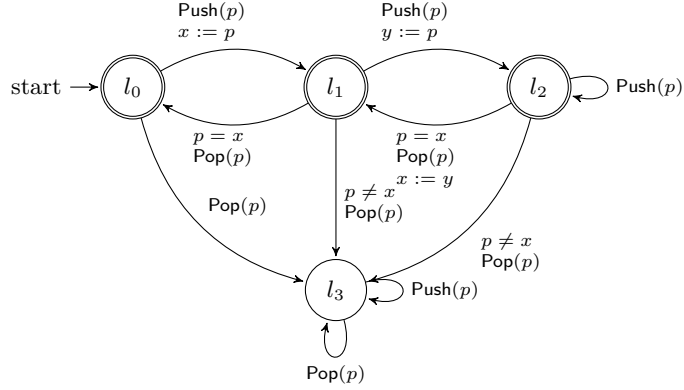
\begin{figure}[htb]
		\centering
		\begin{tikzpicture}[->,>=stealth',shorten >=1pt,auto,node distance=2.7cm,main node/.style={circle,draw,font=\sffamily\large\bfseries}]
		\node[initial, state,accepting] (1) {$l_0$};
		\node[state,accepting] (2) [right of=1] {$l_1$};
		\node[state,accepting] (3) [right of=2] {$l_2$};
		\node[state] (4) [below of=2] {$l_3$};
		
		\path[every node/.style={font=\sffamily\scriptsize}]
		(1) edge [bend left, text width=1.5cm] node {$\push(p)$ \\ $x:=p$} (2)
		    edge [bend right] node {$\pop(p)$} (4)
		(2) edge [bend left, text width=1.5cm] node {$\push(p)$ \\ $y:=p$} (3)
		    edge [bend left, text width=1.1cm] node {$p = x$ \\ $\pop(p)$} (1)
		    edge [text width=2cm] node[pos=0.6] {$p \neq x$ \\ $\pop(p)$} (4)
		(3) edge [bend left, text width=1.1cm] node {$p = x$ \\ $\pop(p)$ \\ $x:=y$} (2)
		    edge [bend left, text width=2cm] node {$p \neq x$ \\ $\pop(p)$} (4)
		    edge [loop right] node {$\push(p)$} (3)
		(4) edge [loop right] node {$\push(p)$} (4)
		    edge [loop below] node {$\pop(p)$} (4);
		\end{tikzpicture}	
		\caption{FIFO-buffer with a capacity of 2 modeled as a register automaton.}
		\label{fig:fifo}
	\end{figure}
\end{example}

\subsection{Semantics of a RA}\label{sec:semanticsRA}
We now formalise the semantics of an RA\@.
A \emph{valuation} of a set of variables $X$ is a function $\nu : X \rightarrow \D$ that assigns data values to variables in $X$.
If $\nu$ is a valuation of $X$ and $g$ is a guard over $X$ then $\nu \models g$ is defined inductively by:
\begin{itemize}
	\item 
	$\nu \models r(x_1,\ldots,x_n)$  iff $(\nu(x_1),\ldots,\nu(x_n)) \in r^{\R}$
	\item 
	$\nu \models \neg r(x_1,\ldots,x_n)$  iff $(\nu(x_1),\ldots,\nu(x_n)) \not\in r^{\R}$
	\item 
	$\nu \models g_1 \wedge g_2$ iff $\nu  \models g_1$ and $\nu \models g_2$
\end{itemize}
A \emph{state} of a RA \( \M = (L, l_0, \mathcal{X}, \Gamma,
\lambda)\) is a pair \(\tuple{l,\nu}\), where \(l \in L\) is a location and \(\nu : \X(l) \xrightarrow[]{} \D \) is a valuation of the set of registers at location \(l\). 
A \emph{run} of \(\M\) over data word \(w = \alpha_1(d_1) \ldots \alpha_n(d_n)\) is a sequence 
\[
  \tuple{l_0, \nu_0} \xrightarrow{\alpha_1(d_1), g_1, \pi_1} \tuple{l_1, \nu_1} \ldots
  \tuple{l_{n-1}, \nu_{n-1}} \xrightarrow{\alpha_n(d_n), g_n, \pi_n} \tuple{l_n, \nu_n},
\]
where 
\begin{itemize}
	\item 
	for each $0 \leq i \leq n$, $\tuple{l_i, \nu_i}$ is a state (with $l_0$ the initial location),
	\item 
	for each $0 < i \leq n$, \(\tuple{l_{i-1}, \alpha_i(p), g_i, \pi_i, l_i} \in \Gamma \) such that
	 \( \iota_i \satisfies g_i \) and  \(\nu_i = \iota_i \circ \pi_i \), where $\iota_i = \nu_{i-1} \cup \{ [p \mapsto d_{i}] \}$ extends \(\nu_{i-1}\) by mapping $p$ to $d_i$.
\end{itemize}
A run is \emph{accepting} if \(\lambda(l_n) = + \), else \emph{rejecting}.
The language of $\M$, notation $L(\M)$, is the set of words $w$ such that $\M$ has an accepting run over $w$.
Word $w$ is \emph{accepted (rejected) under} valuation $\nu_0$ if $\M$ has an accepting (rejecting) run that starts in state $\langle l_0, \nu_o \rangle$.

\begin{example}
	Consider the FIFO-buffer example from~\Cref{fig:fifo}.  This RA has a run
	\begin{eqnarray*}
		\tuple{l_0, \nu_0 = []} & \xrightarrow{\push(7), g_1 \equiv \top, \pi_1 = [x \mapsto p]} & \tuple{l_1, \nu_1 = [x \mapsto 7]}  \\
		& \xrightarrow{\push(7), g_2 \equiv \top, \pi_2 = [x \mapsto x, y \mapsto p]} & \tuple{l_2, \nu_2 = [x\mapsto 7, y \mapsto 7]} \\
		& \xrightarrow{\pop(7), g_3 \equiv p=x, \pi_3 = [x \mapsto y]}& \tuple{l_1, \nu_3 = [x \mapsto 7]}\\
		& \xrightarrow{\push(5), g_4 \equiv \top, \pi_4 = [x \mapsto x, y \mapsto p]} & \tuple{l_2, \nu_4 = [x \mapsto 7, y\mapsto 5]} \\
		& \xrightarrow{\pop(7), g_5 \equiv p=x, \pi_5 = [x \mapsto y]} & \tuple{l_1, \nu_5 = [x \mapsto 5]}\\
		& \xrightarrow{\pop(5), g_6 \equiv p=x, \pi_6 = []} & \tuple{l_0, \nu_6 = []}
	\end{eqnarray*}
	and thus the trace is $\push(7) ~  \push(7)~ \pop(7)  ~ \push(5) ~  \pop(7)  ~ \pop(5)$.
  \fin{}
\end{example}

\subsection{Symbolic Decision Tree\label{sec:SymbolicDecisionTree}}

The $\mathit{SL}^*$ algorithm uses  \emph{tree queries} in place of membership queries.
The arguments of a tree query are a prefix data word $u$ and a symbolic suffix $w$, i.e., a data word with uninstantiated data parameters.
The response to a tree query is a so called \emph{symbolic decision tree} (SDT), which
has the form of tree-shaped register automaton that accepts/rejects suffixes 
obtained by instantiating data parameters in one of the symbolic suffixes.
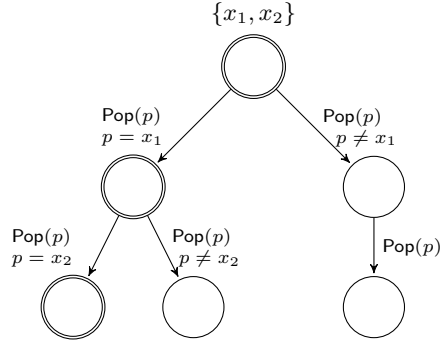
\begin{figure}[t!]
\begin{center}
\begin{tikzpicture}[->,>=stealth',shorten >=1pt,auto,scale=0.8]
\node[state, accepting] (N4) at (0,0) {};
\node[state] (N6) at (5,0){};
\node[state] (N5) at (2,0){};

\node[state, accepting] (N2) at (1,2){};
\node[state] (N3) at (5,2){};

\node[state, accepting, label={$\set{x_1, x_2}$}] (N1) at (3,4){};

\path[every node/.style={font=\sffamily\scriptsize}]
	(N1) edge [text width=1.1cm] node[left] {$\pop(p)$\\$p = x_1$} (N2)
 	     edge [text width=2cm] node[right] {$\pop(p)$\\$~~p \neq x_1$} (N3)
 	(N2) edge [text width=1.1cm] node[left] {$\pop(p)$\\$p = x_2$} (N4)
         edge [text width=1.1cm] node[right] {$\pop(p)$\\$~p \neq x_2$} (N5)
    (N3) edge [text width=1.1cm] node[right] {$\pop(p)$} (N6);
\end{tikzpicture}
\caption{SDT for prefix $\push(5)~ \push(7)$ and (symbolic) suffix $\pop ~ \pop$.\label{fig:SDT}}
\end{center}
\end{figure}
Let us illustrate this on the FIFO-buffer example from~\Cref{fig:fifo}
for the prefix $\push(5)~ \push(7)$ and the symbolic suffix $\pop ~ \pop$.
The acceptance/rejection of suffixes obtained by instantiating data parameters after $\push(5) ~ \push(7)$ can be represented by the SDT in~\Cref{fig:SDT}.
In the initial location, values $5$ and $7$ from the prefix are stored in registers $x_1$ and $x_2$, respectively.
Thus, SDTs will generally not be simple RAs.
Moreover, since the leaves of an SDT have no outgoing transitions, they are also not completely specified.
We use the convention that register $x_i$ stores the $i^{\mathit{th}}$ data value.
Thus, initially, register $x_1$ contains value $5$ and register $x_2$ contains value $7$.
The initial transitions in the SDT contain an update $x_3 := p$, and the final transitions an update $x_4 :=p$. 
For readability, these updates are not displayed in the diagram.
The SDT accepts suffixes of form $\pop(d_1) ~ \pop(d_2)$ iff $d_1$ equals the value stored in register $x_1$, and $d_2$ equals the data value stored in register $x_2$.
\iflong

The formal definitions of an SDT and the notion of a tree oracle are presented in Appendix~\ref{appendix:treeOracleEquality}.
\fi
For a more detailed discussion of SDTs we refer to~\cite{Cassel2016}.



%
\section{Tainting}\label{sec:tainting}
We postulate that the behaviour of the SUT (in our case: a Python program) can be modeled by a register automaton $\mathscr{M}$.
In a black-box setting, observations on the SUT will then correspond to words from the data language of $\mathscr{M}$.
In this section, we will describe the additional observations that a learner can make in a grey-box setting, where the constraints on the data parameters that are imposed within a run become visible.
In this setting, observations of the learner will correspond to what we call tainted words of $\mathscr{M}$.
Tainting semantics is an extension of the standard semantics in which each input value is ``tainted'' with a unique marker from \(\V \).
In a data word \(w = \alpha_1(d_1) \alpha_2(d_2) \ldots \alpha_n(d_n)\), the first data value $d_1$ is tainted with marker $v_1$, the second data value $d_2$ with $v_2$, etc.
While the same data value may occur repeatedly in a data word, all the markers are different.


\subsection{Semantics of Tainting}\label{sec:taintingSemantics}

A \emph{tainted state} of a RA \( \M = (L, l_0, \mathcal{X}, \Gamma,
\lambda)\) is a triple \(\tuple{l,\nu, \zeta}\), where \(l \in L\) is a location, \(\nu : \X(l) \rightarrow \D \) is a valuation, and $\zeta : \X(l) \rightarrow \V$ is a function that assigns a marker to each register of $l$. 
A \emph{tainted run} of \(\M\) over data word \(w = \alpha_1(d_1) \ldots \alpha_n(d_n)\) is a sequence 
\[
\tau = \tuple{l_0, \nu_0, \zeta_0} \xrightarrow{\alpha_1(d_1), g_1, \pi_1} \tuple{l_1, \nu_1, \zeta_1} \ldots
\tuple{l_{n-1}, \nu_{n-1}, \zeta_{n-1}} \xrightarrow{\alpha_n(d_n), g_n, \pi_n} \tuple{l_n, \nu_n, \zeta_n},
\]
where 
\begin{itemize}
	\item 
	$\tuple{l_0, \nu_0} \xrightarrow{\alpha_1(d_1), g_1, \pi_1} \!\tuple{l_1, \nu_1} \ldots
	\!\tuple{l_{n-1}, \nu_{n-1}} \xrightarrow{\alpha_n(d_n), g_n, \pi_n} \!\tuple{l_n, \nu_n}$ is a run of \(\mathscr{M}\),
	\item 
	for each $0 \leq i \leq n$, $\tuple{l_i, \nu_i, \zeta_i}$ is a tainted state,
	\item 
	for each $0 < i \leq n$, $\zeta_i = \kappa_i \circ \pi_i$, where $\kappa_i = \zeta_{i-1} \cup \{ (p,v_i) \}$.
\end{itemize}
The tainted word of $\tau$ is the sequence $w = \alpha_1 (d_1) G_1 \alpha_2 (d_2) G_2 \cdots \alpha_n (d_n) G_n$,
where $G_i = g_i[\kappa_i]$, for $0 < i \leq n$. 
We define $\constraints_{\M}(\tau) = [G_1,\ldots, G_n]$. 

Let $w = \alpha_1(d_1) \ldots \alpha_n(d_n)$ be a data word.  Since register automata are deterministic, there is a unique tainted run $\tau$ over $w$.  We define $\constraints_{\M}(w) = \constraints_{\M}(\tau)$, that is, the constraints associated to a data word are the constraints of the unique tainted run that corresponds to it.
In the untainted setting a membership query for data word $w$ leads to a response ``yes'' if $w \in L(\M)$, and a response ``no'' otherwise, but in a tainted setting the predicates $\constraints_{\M}(w)$ are also included in the response, and provide additional information that the learner may use.

\begin{example}
Consider the FIFO-buffer example from~\Cref{fig:fifo}.  This RA has a tainted run
\begin{alignat*}{2}
	 \tuple{l_0, [],[]} & \xrightarrow{\push(7)} && \tuple{l_1, [x \mapsto 7], [x\mapsto v_1]}
 \xrightarrow{\push(7)}  \tuple{l_2, [x\mapsto 7,y\mapsto 7], [x\mapsto v_1, y \mapsto v_2]} \\
& \xrightarrow{\pop(7)}&& \tuple{l_1, [x\mapsto 7], [x\mapsto v_2]}
 \xrightarrow{\push(5)}  \tuple{l_2, [x\mapsto 7,y\mapsto 5], [x\mapsto v_2, y \mapsto v_4]} \\
& \xrightarrow{\pop(7)} && \tuple{l_1, [x\mapsto 5], [y\mapsto v_4]}
 \xrightarrow{\pop(5)}  \tuple{l_0, [], []}
\end{alignat*}
(For readability, guards $g_i$ and assignments $\pi_i$ have been left out.)
The constraints in the corresponding tainted trace can be computed as follows:
\begin{align*}
&\kappa_1 = [p \mapsto v_1]  && G_1 \equiv \top[\kappa_1] \equiv \top\\
&\kappa_2 = [x \mapsto v_1, p \mapsto v_2]  && G_2 \equiv \top[\kappa_2] \equiv \top\\
&\kappa_3 = [x \mapsto v_1, y \mapsto v_2, p \mapsto v_3]  && G_3 \equiv (p=x)[\kappa_3] \equiv v_3 = v_1\\
&\kappa_4 = [x \mapsto v_2, p \mapsto v_4]  && G_4 \equiv \top[\kappa_4] \equiv \top\\
&\kappa_5 = [x \mapsto v_2, y \mapsto v_4, p \mapsto v_5]  && G_5 \equiv (p=x)[\kappa_5] \equiv v_5 = v_2\\
&\kappa_6 = [x \mapsto v_4, p \mapsto v_6]  && G_6 \equiv (p=x)[\kappa_6] \equiv v_6 = v_4
\end{align*}
and thus the tainted word is:
\[
\push(7) ~ \top~  \push(7) ~ \top ~ \pop(7) ~ v_3 = v_1 ~ \push(5) ~ \top~  \pop(7) ~ v_5=v_2 ~ \pop(5) ~ v_6 = v_4,
\]
and the corresponding list of constraints is $ [ \top, \top, v_3 = v_1, \top, v_5 = v_2, v_6 = v_4]$.
\fin{}
\end{example}

Various techniques can be used to observe tainted traces, for instance symbolic and concolic execution.
In this work, we have used a library called ``\texttt{taintedstr}'' to achieve tainting in Python and make tainted traces available to the learner.

\subsection{Tainting in Python}\label{sec:PythonTainting}

Tainting in Python is achieved by using a library called
``\texttt{taintedstr}''\footnote{See \cite{abs-1810-08289} and
  \url{https://github.com/vrthra/taintedstr}.}, which implements a
``\texttt{tstr}'' (\emph{tainted string}) class.
We do not discuss the entire implementation in detail, but only introduce the
portions relevant to our work.
The ``\texttt{tstr}'' class works by \emph{operator overloading}:
each operator is overloaded to record its own invocation.
The \texttt{tstr} class overloads the implementation of the ``\texttt{\_\_eq\_\_}''
(equality) method in Python's \texttt{str} class, amongst others.
In this text, we only consider the equality method. 
A \texttt{tstr} object \(x\) can be considered as a triple \(\tuple{o, t, \mathit{cs}}\),
where \(o\) is the (base) string object, \(t\) is the taint value associated
with string \(o\), and \(\mathit{cs}\) is a set of comparisons made by \(x\) with other
objects, where each comparison \(c \in \mathit{cs}\) is a triple \(\tuple{f, a, b}\) with
\(f\) the name of the binary method invoked on \(x\), \(a\) a copy of \(x\), and
\(b\) the argument supplied to \(f\).

Each a method \(f\) in the \texttt{tstr} class is an overloaded implementation
of the relevant (base) method \(f\) as follows:
\begin{lstlisting}[style=py, language=Python]
def f(self, other):
    self.cs.add((m._name_, self, other))
    return self.o.f(other) # `o' is the base string
\end{lstlisting}
We present a short example of how such an overloaded method would work below:
\begin{example}[\texttt{tstr} tainting]\label{example:tstrTainting}
  Consider two \texttt{tstr} objects: \(x_1 = \tuple{ \textnormal{``1''} , 1, \emptyset}\)
  and \(x_2 = \tuple{ \textnormal{``1''}, 2, \emptyset}\). 
  Calling \(x_1 == x_2 \) returns \textbf{True} as $x_{1}.o = x_{2}.o$.
  As a side-effect of $f$, the set of comparisons $x_{1}.cs$ is updated with the
  triple \(c = \tuple{ \textnormal{``\_\_eq\_\_''}, x_1, x_2}\).
  We may then confirm that \(x_1\) is compared to \(x_2\) by checking the taint
  values of the variables in comparison \(c\): \(x_1.t = 1\) and \(x_2.t
  = 2\).

 Note, our approach to tainting limits the recorded information to operations performed
on a \texttt{tstr} object.
\iflong
\fin{}
\end{example}
\begin{example}[Complicated Comparison]
\fi
  Consider the following snippet, where \(x_1, x_2, x_3\) are \texttt{tstr} objects
with \(1,2,3\) as taint values respectively:
\begin{lstlisting}[style=py, language=Python, gobble=2]
  if not (x_1 == x_2 or (x_2 != x_3)):
      # do something
\end{lstlisting}
  If the base values of \(x_1\) and \(x_2\) are equal, the
  Python interpreter will ``short-circuit'' the if-statement and the second
  condition, \(x_2 \neq x_3\), will not be evaluated.
  Thus, we only obtain one comparison: $x_{1} = x_{2}$.
  On the other hand, if the base values of \(x_1\) and \(x_2\) are not equal,
  the interpreter will not short-circuit, and both comparisons will be recorded
  as \( \set{x_2 = x_3, x_1 \neq x_2}\).
  \iflong
  While the comparisons are stored as a set, from the perspective of the tainted
  trace, the guard(s) is a single conjunction: \( x_2 = x_3 \wedge x_1 \neq x_2
  \).
  \fi
  However, the external negation  operation
  will not be recorded by any of the \texttt{tstr} objects: the
  negation was not performed on the \texttt{tstr} objects.
  \fin{}
\end{example}

\section{Learning Register Automata using Tainting}\label{sec:learningWithTainting}

Given an SUT and a tree query, we generate an SDT in the following steps:
\emph{(i)} construct a \emph{characteristic predicate} of the tree query (\Cref{alg:decisionQuery}) using membership and guard queries,
\emph{(ii)} transform the characteristic predicate into an SDT (\Cref{alg:fullSDT}), and
\emph{(iii)} minimise the obtained SDT (\Cref{alg:minimalSDT}).

\subsection{Tainted Tree Oracle \label{sec:taintedTreeOracle}}

\subsubsection{Construction of Characteristic Predicate}\label{sec:decisionRelation}

For $u = \alpha(d_1) \cdots \alpha_k(d_k)$ a data word, $\nu_u$ denotes the valuation of $\set{x_1,\ldots,x_k}$ with $\nu_u(x_i) = d_i$, for $1 \leq i \leq k$.
Suppose $u$ is a prefix and $w= \alpha_{k+1} \cdots \alpha_{k+n}$ is a symbolic suffix. Then $H$ is a \emph{characteristic predicate} for $u$ and $w$ in $\M$ if, for each valuation $\nu$ of $\set{x_1,\ldots, x_{k+n}}$ that extends $\nu_u$,
\[
	\nu \models H \iff  \alpha_1(\nu(x_1)) \cdots \alpha_{k+n}(\nu(x_{k+n})) \in L(\M),
\]
that is, $H$ characterizes the data words $u'$ with $\Acts(u') = w$ such that $u \cdot u'$ is accepted by $\M$.
In the case of the FIFO-buffer example from~\Cref{fig:fifo}, a characteristic predicate
for prefix $\push(5)~ \push(7)$ and symbolic suffix $\pop ~ \pop$ is $x_3 = x_1 \wedge x_4 = x_2$.
A characteristic predicate for the empty prefix and symbolic suffix $\pop$ is $\perp$, since this trace will inevitably lead to the sink location $l_3$ and there are no accepting words.

\Cref{alg:decisionQuery} shows how a characteristic predicate may be computed by systematically exploring all the (finitely many) paths of $\M$ with prefix $u$ and suffix $w$ using tainted membership queries.
During the execution of~\Cref{alg:decisionQuery}, predicate $G$ describes the part of the parameter space that still needs to be explored, whereas $H$ is the characteristic predicate for the part of the parameter space that has been covered.
We use the notation $H \equiv T$ to indicate syntactic equivalence, and $H = T$ to indicate logical equivalence.
Note, if there exists no parameter space to be explored (i.e., $w$ is empty)
and $u \in L(\M)$, the algorithm returns $H \equiv \perp \vee \top$ (as the empty conjunction equals $\top$).
\begin{algorithm}[t]
    \KwData{A tree query consisting of prefix \(u = \alpha_1(d_1) \cdots \alpha_k(d_k)\) and symbolic suffix \(w = \alpha_{k+1} \cdots \alpha_{k+n} \)}
    \KwResult{A characteristic predicate for $u$ and $w$ in $\M$}
    \(G := \top \), \(H := \bot \),
    $V :=  \set{x_1,\ldots, x_{k+n}}$ \;
    \While{$\exists \mbox{ valuation } \nu \mbox{ for }  V \mbox{ that extends } \nu_u \mbox{ such that } \nu \models G$}{
        \(\nu := \) valuation for $V$ that extends \(\nu_u \) such that \(\nu \models G \)  \;
        \(z :=  \alpha_1(\nu(x_1)) \cdots \alpha_{k+n}(\nu(x_{k+n})) \) 
        \tcp*{Construct membership query}
        \(I := \bigwedge_{i=k+1}^{k+n} \constraints_{\M}(z)[i]\)
        \tcp*{Constraints resulting from query}
        \uIf(\tcp*[f]{Result query ``yes'' or ``no''}){$z \in L(\M)$}{
          \(H :=H \vee I\)
        }
        \(G := G \wedge \neg I \) \;
    }
    \textbf{return} \(H \)
    \caption{\texttt{ComputeCharacteristicPredicate}}\label{alg:decisionQuery}
\end{algorithm}
\begin{example}[\Cref{alg:decisionQuery}]
	Consider the FIFO-buffer example and the tree query with prefix $\push(5)~ \push(7)$ and symbolic suffix $\pop ~ \pop$.
	After the prefix location $l_2$ is reached.  From there, three paths are possible with actions $\pop ~ \pop$:
	$l_2 l_3 l_3$, $l_2 l_1 l_3$ and $l_2 l_1 l_0$. We consider an example run of \Cref{alg:decisionQuery}.
	
	Initially, $G_0 \equiv \top$ and $H_0 \equiv \perp$.  Let
	$\nu_1 = [ x_1 \mapsto 5, x_2 \mapsto 7, x_3 \mapsto 1, x_4 \mapsto 1]$.
	Then $\nu_1$ extends $\nu_u$ and $\nu_1 \models G_0$. The resulting tainted run corresponds to path $l_2 l_3 l_3$ and so the tainted query gives path constraint $I_1 \equiv x_3 \neq x_1 \wedge \top$. Since the tainted run is rejecting, $H_1 \equiv \perp$ and $G_1 \equiv \top \wedge \neg I_1$.
	
	In the next iteration, we set $\nu_2 = [ x_1 \mapsto 5, x_2 \mapsto 7, x_3 \mapsto 5, x_4 \mapsto 1]$.
	Then $\nu_2$ extends $\nu_u$ and $\nu_2 \models G_1$. The resulting tainted run corresponds to path $l_2 l_1 l_3$ and so the tainted query gives path constraint $I_2 \equiv x_3 = x_1 \wedge x_4 \neq x_2$. Since the tainted run is rejecting, $H_2 \equiv \perp$ and $G_2 \equiv \top \wedge \neg I_1 \wedge \neg I_2$.
	
	In the final iteration, we set $\nu_3 = [ x_1 \mapsto 5, x_2 \mapsto 7, x_3 \mapsto 5, x_4 \mapsto 7]$.
	Then $\nu_3$ extends $\nu_u$ and $\nu_3 \models G_2$. The resulting tainted run corresponds to path $l_2 l_1 l_0$ and the tainted query gives path constraint $I_3 \equiv x_3 = x_1 \wedge x_4 = x_2$. Now the tainted run is accepting, so $H_3 \equiv \perp \vee I_3$ and $G_3 = \top \wedge \neg I_1 \wedge \neg I_2 \wedge \neg I_3$. As $G_3$ is unsatisfiable, the algorithm terminates and returns characteristic predicate $H_3$.
\end{example}

\subsubsection{Construction of a non-minimal SDT\label{sec:nonMinimalSDT}}
For each tree query with prefix $u$ and symbolic suffix $w$, the corresponding characteristic predicate $H$ is sufficient to construct an SDT using~\Cref{alg:fullSDT}.

\begin{algorithm}[h]
  \KwData{Characteristic predicate $H$, index $n = k+1$,
    \\Number of suffix parameters $N$}
    \KwResult{Non-minimal SDT \(\tree \)}
    \uIf{\( n = k+N +1\)}{
        \(l_0 :=  \) SDT node\;
        \(z :=\) if $H \iff \bot $ then $-$ else $+$
        \tcp*{Value $\lambda$ for leaf node of the SDT}
        \textbf{return} \(\tuple{\set{l_0}, l_0, [l_0 \mapsto \emptyset], \varnothing, [l_0 \mapsto z]} \)
        \tcp*{RA with single location}
    }
    \uElse{
        \(\tree :=  \) SDT node\;
        \(I_t := \set{i \mid x_{n} \odot x_{i} \in H, \, n > i} \)
        \tcp*{$x_{i}$ may be a parameter or a constant}
        \uIf{\(I_t \) is \(\varnothing \)}{
          \(t :=  \) \texttt{SDTConstructor}\(\parens{H, n+1, N} \)
          \tcp*{No guards present}
          Add \(t \) with guard $\top$ to \(\tree \) \;
        }
        \uElse{
            \(g := \bigwedge_{i \in I_t} x_{n} \neq x_{i} \) \tcp*{Disequality guard case}
            \(H' := \bigvee_{f \in H} f \wedge g \) if $f \wedge g$ is
            satisfiable else $\bot$ \tcp*{$f$ is a disjunct}
            \(t' :=  \) \texttt{SDTConstructor}\(\parens{H', n+1, N} \) \;
            Add \(t' \) with guard $g$ to \(\tree \) \;
            \For{\(i \in I_t \)}{
              \(g := x_{n} = x_{i} \) \tcp*{Equality guard case}
              \(H' := \bigvee_{f \in H} f \wedge g \) if $f \wedge g$ is
            satisfiable else $\bot$ \;
                \(t':=\) \texttt{SDTConstructor}\(\parens{H', n+1, N} \) \;
                Add \(t' \) with guard $g$ to \(\tree \)
            }
        }
        \textbf{return} \(\tree \)
    }
    \caption{\texttt{SDTConstructor}}\label{alg:fullSDT}
\end{algorithm}
\iflong

\Cref{alg:fullSDT} proceeds in the following manner: for a symbolic action
\(\alpha\parens{x_{n}} \) with parameter \(x_{n} \), construct the \emph{potential set} \(I_t \) (lines 6 \& 7), that is,
the set of parameters to which $x_{n}$ is compared to in $H$.
For line $7$, recall that $H$ is a DNF formula, hence each literal $x_{j} \odot x_{k}$ is considered in the set comprehension, rather than the conjunctions making up the predicate $H$.
Each element $x_{i} \in I_{t}$ can be either a formal parameter in the tree query or a constant $c_{i} \in C$ from our chosen structure.
Using \(I_t \), we can construct the guards as follows:
\begin{itemize}
    \item \textbf{Disequality guard}: The disequality guard will be \(g :=
      \bigwedge_{\set{i \in I_t}} x_{n} \neq x_{i} \).
    We can then check which guards in $H$ are still satisfiable with the
    addition of $g$ and constructs the predicate $H'$ for the next call of
    \Cref{alg:fullSDT} (lines 13--16).

  \item \textbf{Equality guard (s)}: For each parameter $x_{i}$ for
    \(i \in I_t \), the equality guard will be \(g := x_{n} = x_{i} \).
    We can then check which guards in $H$ are still satisfiable with the
    addition of $g$ and this becomes the predicate $H'$ for the next call of
     \Cref{alg:fullSDT} (lines 18--21).
\end{itemize}

At the base case (lines $1-4$), there are no more parameters remaining
and we return a non-accepting leaf if $H = \bot$, otherwise accepting.
As mentioned, at each non-leaf location $l$ of the SDT $\tree$ returned
by~\Cref{alg:fullSDT}, there exists a potential set \(I_t\).
For each parameter $x_{i}$, we know that there is a comparison between
$x_i$ and $x_{n}$ in the SUT.

\fi

\begin{example}[\Cref{alg:fullSDT}\label{example:fullSDT}]
  Consider a characteristic predicate $H \equiv I_{1} \vee I_{2} \vee I_{3} \vee I_{4}$, where
  $I_{1} \equiv x_{2} \neq x_{1} \wedge x_{3} \neq x_{1}$,
  $I_{2} \equiv x_{2} = x_{1} \wedge x_{3} \neq x_{1}$,
  $I_{3} \equiv x_2 \neq x_{1} \wedge x_3 = x_{1}$,
  $I_{4} \equiv x_2 = x_{1} \wedge x_3 = x_{1}$.
  We discuss only the construction of the sub-tree rooted at node $s_{21}$
  for the SDT visualised in~\Cref{fig:nonMinimalSDT}; the construction of the
  remainder is similar.

  Initially, $x_{n} = x_{k+1} = x_{2}$.
  Potential set $I_{t}$ for $x_{2}$ is $\set{x_{1}}$ as $H$ contains the literals
  $x_{2} = x_{1}$ and $x_{2} \neq x_{1}$.
  Consider the construction of the equality guard $g := x_{2} = x_{1}$.
  The new characteristic predicate is
  $H' \equiv (I_{2} \wedge g) \vee (I_{4} \wedge g)$, as
  $I_{1}$ and $I_{3}$ are unsatisfiable when conjugated with $g$.

  For the next call, with $n=3$, the current variable is $x_{3}$,
  with predicate $H = H'$ (from the parent instance).
  We obtain the potential set for $x_{3}$ as $\set{x_{1}}$.
  The equality guard is $g' := x_{3} = x_{1}$ with the new characteristic predicate
  $H'' \equiv I_{4} \wedge g \wedge g'$, i.e., $H'' \iff x_{2} = x_{1} \wedge x_{3} = x_{1}$
  (note, $I_{2} \wedge g \wedge g'$ is unsatisfiable).
  In the next call, we have $n = 4$, thus we compute a leaf.
  As $H''$ is not $\bot$, we return an accepting leaf $t$.
  The disequality guard is $g'' := x_{3} \neq x_{1}$ with characteristic predicate
  $H''' \iff x_{2} = x_{1} \wedge x_{3} = x_{1} \wedge x_{3} \neq x_{1} \iff \bot$.
  In the next call, we have $n=4$, and we return a non-accepting leaf $t'$.
  The two trees $t$ and $t'$ are added as sub-trees with their respective guards $g'$ and $g''$
  to a new tree rooted at node $s_{21}$ (see~\Cref{fig:nonMinimalSDT}).
\fin{}
\end{example}


\subsubsection{SDT Minimisation}\label{sec:SDTminimisation}
\Cref{example:fullSDT} showed a characteristic predicate $H$ containing redundant comparisons, resulting in the non-minimal SDT in~\Cref{fig:nonMinimalSDT}.
We use~\Cref{alg:minimalSDT} to minimise the SDT in~\Cref{fig:nonMinimalSDT} to the SDT in~\Cref{fig:minimalSDT}.

\begin{algorithm}[h]
    \KwData{Non-minimal SDT \(\tree \), current index $n$}
    \KwResult{Minimal SDT \(\tree' \)}
    \uIf(\tcp*[h]{Base case}){\(\tree \) is a leaf}{
        \textbf{return} \(\tree \)
    }
    \uElse{
        \(\tree' :=  \) SDT node\\
        \tcp{Minimise the lower levels}
        \For{guard \(g \) with associated sub-tree \(t \) in \(\tree \)}{
            Add guard \(g \) with associated sub-tree \(\texttt{MinimiseSDT}(t, n+1) \) to \(\tree' \)
        }
        \tcp{Minimise the current level}
        \(I :=  \) Potential set of root node of \(\tree \)\\
        \(t' :=  \) disequality sub-tree of \(\tree \) with
        guard $\bigwedge_{i \in I} x_{n} \neq x_{i}$\\
        \(I' := \varnothing \)\\
        \For{\(i \in I \)}{
            \(t :=  \) sub-tree of \(\tree \) with guard \(x_{n} = x_{i} \)\\
            \uIf{\(t' \tuple{x_{i},x_{n}} \not\iso t  \) or
              \(t' \tuple{x_{i}, x_{n}}  \) is undefined}{
                \(I' := I' \cup \set{x_i} \)\\
                Add guard \(x_{n} = x_{i} \) with corresponding sub-tree
                \(t \) to \(\tree' \)
            }
        }
        Add guard \(\bigwedge_{i \in I'} x_{n} \neq x_{i}  \) with
        corresponding sub-tree \(t' \) to \(\tree' \)\\
        \textbf{return} \(\tree' \)
    }
    \caption{\texttt{MinimiseSDT}}\label{alg:minimalSDT}
\end{algorithm}

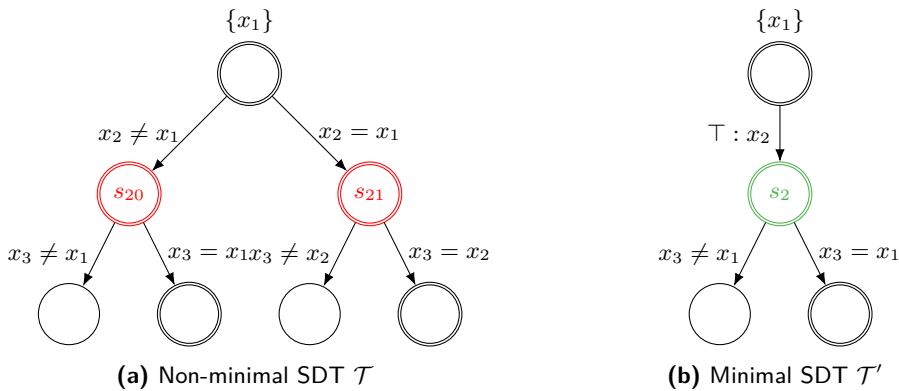
\begin{figure}[h]
  \centering
  \subcaptionbox{Non-minimal SDT $\tree$ \label{fig:nonMinimalSDT}}{\begin{tikzpicture}[scale=0.8]

\node[state] (rejecting1) at (0,0) {};
\node[state] (rejecting2) at (4,0) {};
\node[state, accepting] (accepting1) at (2,0) {};
\node[state, accepting] (accepting2) at (6,0) {};

\node[state, accepting, color=red] (s1_neq_r0) at (1,2) {$s_{20}$};
\node[state, accepting, color=red] (s1_eq_r0) at (5,2) {$s_{21}$};
\node[state, accepting, label=$\set{x_{1}}$] (root) at (3,4) {};

\path[-Latex] (root) edge [] node[left] {\(x_2 \neq x_1\)} (s1_neq_r0);
\path[-Latex] (root) edge [] node[right] {\(x_2 = x_1\)} (s1_eq_r0);

\path[-Latex] (s1_neq_r0) edge [] node[left] {\(x_3 \neq x_1\)} (rejecting1);
\path[-Latex] (s1_eq_r0) edge [] node[left] {\(x_3 \neq x_2\)} (rejecting2);

\path[-Latex] (s1_eq_r0) edge [] node[right] {\(x_3 = x_2\)} (accepting2);
\path[-Latex] (s1_neq_r0) edge [] node[right] {\(x_3 = x_1\)} (accepting1);
  
\end{tikzpicture}

}%
  \hfill
  \subcaptionbox{Minimal SDT $\tree'$ \label{fig:minimalSDT}}{\begin{tikzpicture}[scale=0.8]

\node[state] (rejecting1) at (0,0) {};
\node[state, accepting] (accepting1) at (2,0) {};

\node[state, accepting, color=green!40!gray] (s2_root) at (1,2) {$s_{2}$};
\node[state, accepting, label=$\set{x_{1}}$] (root) at (1,4) {};

\path[-Latex] (root) edge [] node[left] {$\top : x_2$} (s2_root);

\path[-Latex] (s2_root) edge [] node[left] {\(x_3 \neq x_1\)} (rejecting1);
\path[-Latex] (s2_root) edge [] node[right] {\(x_3 = x_1\)} (accepting1);
  
\end{tikzpicture}

}%
  \hfill
  \caption{SDT Minimisation:
    Redundant nodes (in red, left SDT) are merged together (in green, right SDT).
  \label{fig:minimisationExample}}
\end{figure}

We present an example of the application of~\Cref{alg:minimalSDT}, shown for the
SDT of~\Cref{fig:nonMinimalSDT}.
\Cref{fig:nonMinimalSDT} visualises a non-minimal SDT $\tree$, where \(s_{20}\) and
\(s_{21}\) (in red) are essentially ``duplicates'' of each other: the sub-tree for node
\(s_{20}\) is isomorphic to the sub-tree for node \(s_{21}\) under the relabelling
``\(x_2 = x_1\)''.
We indicate this relabelling using the notation $\tree[s_{20}]\tuple{x_1, x_2}$ and the
isomorphism relation under the relabelling as $\tree[s_{20}]\tuple{x_1, x_2} \iso \tree[s_{21}]$.
\Cref{alg:minimalSDT} accepts the non-minimal SDT of~\Cref{fig:nonMinimalSDT}
and produces the equivalent minimal SDT in~\Cref{fig:minimalSDT}.
Nodes \(s_{20}\) and \(s_{21}\) are merged into one node, \(s_2\), marked in
green.
We can observe that both SDTs still encode the same decision tree.
With~\Cref{alg:minimalSDT}, we have completed our tainted tree oracle, and can
now proceed to the tainted equivalence oracle.

\subsection{Tainted Equivalence Oracle\label{sec:taintedEquivalenceOracle}}

The \emph{tainted equivalence oracle} (TEO), similar to its non-tainted counterpart, accepts
a hypothesis $\mathcal{H}$ and verifies whether $\mathcal{H}$ is equivalent to
register automaton $\mathcal{M}$ that models the SUT\@.
If $\mathcal{H}$ and $\mathcal{M}$ are equivalent, the oracle replies ``yes'',
otherwise it returns ``no'' together with a CE.
The RandomWalk Equivalence Oracle in RALib constructs random traces in order to find a CE\@.

\begin{definition}[Tainted Equivalence Oracle]\label{def:taintedEquivalenceOracle}
 For a given hypothesis $\mathcal{H}$, maximum word length $n$, and an SUT $\mathcal{S}$, a tainted
 equivalence oracle is a function $\mathcal{O}_{\mathcal{E}}(\mathcal{H}, n,
 \mathcal{S})$ for all tainted traces $w$ of $\mathcal{S}$ where $| w | \leq n$,
 \( \oracle_{\mathcal{E}} (\mathcal{H}, n, \mathcal{S}) \) returns $w$ if \( w
 \in \lang(\mathcal{H}) \iff w \in \lang(\mathcal{S}) \) is false, and `Yes' otherwise.
 \end{definition}

The TEO is similar to the construction of the characteristic predicate to find a CE: we randomly generate a symbolic suffix of specified length $n$ (with an empty prefix), and construct a predicate $H$ for the query.
For each trace $w$ satisfying a guard in $H$, we confirm whether $w \in \lang(\mathcal{H}) \iff w \in \lang(\mathcal{M})$.
If false, $w$ is a CE\@.
If no $w$ is false, then we randomly generate another symbolic suffix.
In practise, we bound the number of symbolic suffixes to generate.
\Cref{example:combinationLock} presents a scenario of a combination lock
automaton that can be learned (relatively easily) using a TEO but cannot be handled by normal oracles.
\begin{example}[Combination Lock RA]\label{example:combinationLock}
  A combination lock is a type of RA which requires a \emph{sequence} of
  specific inputs to `unlock'.
  \begin{figure}[tbp]
    \centering
    \begin{tikzpicture}[->,>=stealth',shorten >=1pt,auto]
  \node[initial, state, accepting] (l_0) at (0,0) {\(l_0\)};
  \node[state, accepting] (l_1) at (2,0) {\(l_1\)};
  \node[state, accepting] (l_2) at (4,0) {\(l_2\)};
  \node[state, accepting] (l_3) at (6,0) {\(l_3\)};
  \node[state, accepting] (l_4) at (8,0) {\(l_4\)};

  \path[->] (l_0) edge [bend left=45] node[above] {\(\frac{\mathit{\alpha(p)} ~\mid~ p = 1}{\emptyset}\)} (l_1);

  \path[->] (l_1) edge [bend left=45] node[above] {\(\frac{\mathit{\alpha(p)} ~\mid~ p = 9}{\emptyset}\)} (l_2);

  \path[->] (l_2) edge [bend left=45] node[above] {\(\frac{\mathit{\alpha(p)} ~\mid~ p = 6}{\emptyset}\)} (l_3);
  
  \path[->] (l_3) edge [bend left=45] node[above] {\(\frac{\mathit{\alpha(p)} ~\mid~ p = 2}{\emptyset}\)} (l_4);

  \path[->] (l_4) edge [loop right] node[above] {$\beta$}(l_4);  

  \path[->] (l_0) edge [loop above] node[left] {\(\frac{\alpha(p) ~\mid ~ p \neq 1}{\emptyset} \)} (l_0);
  
  \draw[->] (l_1) -- ++(0,-1) -| (l_0) node[pos=0.25, below] {\(\frac{\alpha(p) ~\mid ~ p \neq 9}{\emptyset}\)};

  \draw[-] (l_2) -- ++(0,-1) -| ++(-2,0) node[pos=0.25, below] {\(\frac{\alpha(p) ~\mid ~ p \neq 6}{\emptyset}\)};

  \draw[-] (l_3) -- ++(0,-1) -| ++(-2,0) node[pos=0.25, below] {\(\frac{\alpha(p) ~\mid ~ p \neq 2}{\emptyset}\)};
  
\end{tikzpicture}

    \caption{Combination Lock \(\mathcal{C}\) : Sequence
      \(\alpha(1)\alpha(9)\alpha(6)\alpha(2)\) \emph{unlocks} the
      automaton.
    Error transitions (from \(l_3\) -- \(l_1\) to \(l_0\)) have been `merged' for conciseness.
    The sink state has not been drawn.\label{fig:combinationLock}
  }
  \end{figure}
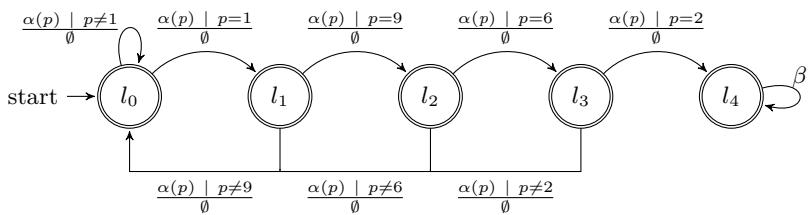
  \Cref{fig:combinationLock} presents an RA \(\mathcal{C}\) with a `4-digit'
  combination lock that can be unlocked by the
  sequence \( w = \alpha(c_0)\alpha(c_1)\alpha(c_2)\alpha(c_3)\), where \(\set{c_0, c_1, c_2, c_3}\) are constants.
  Consider a case where a hypothesis \(\mathcal{H}\) is being checked for
  equivalence against the RA \(\mathcal{C}\) with \(w \not\in
  \lang(\mathcal{H})\).
  While it would be difficult for a normal equivalence oracle to generate the word
  \(w\) randomly; the tainted equivalence oracle will record at every location
  the comparison of input data value \(p\) with some constant \(c_i\) and explore
  all corresponding guards at the location, eventually constructing the word
  \(w\). 

  For the combination lock automaton, we may note that as the `depth' of the lock increases, the possibility of randomly finding a CE decreases.
  \fin{}
\end{example}


%
\section{Experimental Evaluation}\label{sec:experimentalEval}
We have used stubbed versions of the Python FIFO-Queue and Set modules\footnote{From Python's
  \texttt{queue} module and standard library, respectively.} for learning the FIFO and Set models, while the Combination Lock automata were constructed manually.
Source code for all other models was obtained by translating existing benchmarks from~\cite{Neider2019} (see also \url{automata.cs.ru.nl}) to Python code.
We also utilise a `reset' operation:
A `reset' operation brings an SUT back to its initial state, and is counted as an `input' for our purposes.
Furthermore, each experiment was repeated 30 times with different random seeds.
Each experiment was bounded according to the following constraints:
learning phase: \(10^9\) inputs and \(5 \times 10^7\) resets;
testing phase: \(10^9\) inputs and \(5 \times 10^4\) resets;
length of the longest word during testing: 50; and
a ten-minute timeout for the learner to respond.

\Cref{fig:benchmarkPlots} gives an overview of our experimental results.
We use the notation `TTO' to represent `Tainted Tree Oracle' (with similar labels for the other
oracles).
In the figure, we can see that as the size of the container increases, the
difference between the fully tainted version (TTO+TEO, in
blue) and the completely untainted version (NTO+NEO, in red) increases.
In the case where only a tainted tree oracle is used (TTO+NEO, in green), we see
that it is following the fully tainted version closely (for the FIFO models) and
is slightly better in the case of the SET models.

\begin{figure}[h]
  \centering
  \includegraphics[scale=0.7]{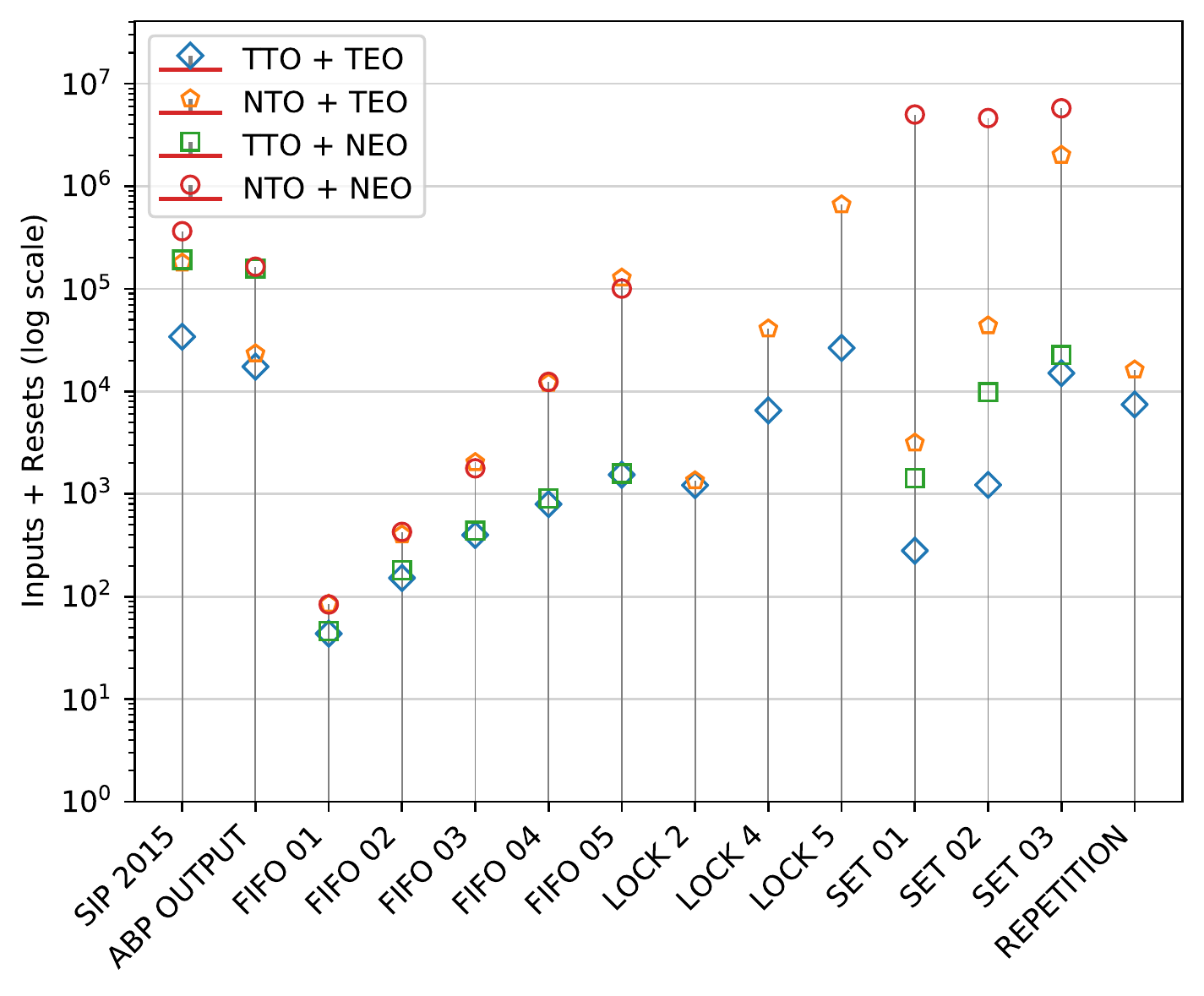}
  \caption{Benchmark plots: Number of symbols used with tainted
    oracles (blue and green) are generally \emph{lower} than with normal oracles
    (red and orange).
    Note that the y-axis is log-scaled.
    Additionally, normal oracles are unable to learn the Combination Lock and
    Repetition automata and are hence not plotted.
  }%
  \label{fig:benchmarkPlots}
\end{figure}

The addition of the TEO gives
a conclusive advantage for the Combination Lock and Repetition benchmarks.
The addition of the TTO by itself results in
significantly fewer number of symbols, even without the tainted equivalence
oracle (TTO v/s NTO, compare the green and red lines).
With the exception of the Combination Lock and Repetition benchmarks,
the TTO+TEO combination does not provide vastly better results in comparison to
the TTO+NEO results, however, it is still (slightly) better.
We note that --- as expected --- the NEO does not manage to provide CEs for the
Repetition and Combination Lock automata.
The TEO is therefore much more useful for finding CEs in SUTs which utilise
constants.
\iflong
For complete details of the data used to produce the plots, please refer to
Appendix~\ref{appendix:benchmarkTable}.
\fi




%

\section{Conclusions and Future Work}\label{sec:conclusion}

In this article, we have presented an integration of 
dynamic taint analysis, a white-box technique for tracing data flow, 
and register automata learning, a black-box technique for
inferring behavioral models of components. The combination 
of the two methods improves upon the state-of-the-art in 
terms of class of systems for which models can be generated 
and in terms of performance:
Tainting makes it possible to infer data-flow constraints even 
in instances with a high essential complexity (e.g., in the 
case of so-called combination locks).
Our implementation outperforms pure black-box learning by
two orders of magnitude with a growing impact in the
presence of multiple data parameters and registers.
Both improvements are important steps towards the 
applicability of model learning in practice
as they will help scaling to industrial use cases.



%
At the same time our evaluation shows the need for further
improvements:
Currently, the $\mathit{SL}^*$ algorithm uses symbolic decision trees and tree 
queries globally, a well-understood weakness of learning algorithms
that are based on observation tables. It also uses individual
tree oracles each type of operation and relies on  
syntactic equivalence of decision trees. A more advanced learning 
algorithm for extended finite state machines will 
be able to consume fewer tree queries, leverage semantic 
equivalence of decision trees. Deeper integration with 
white-box techniques could enable the analysis of many 
(and more involved) operations on data values.



\paragraph{Acknowledgement}
We are grateful to Andreas Zeller for explaining the use of tainting for dynamic tracking of constraints, and to Rahul Gopinath for helping us with his library for tainting Python programs.
We also thank the anonymous reviewers for their suggestions.

\iflong
\bibliographystyle{splncsnat}
\bibliography{reference,abbreviations,dbase}
\fi

\iflong
\clearpage
\begin{subappendices}
\renewcommand{\thesection}{\Alph{section}}%


\section{Tree Oracle for Equalities}\label{appendix:treeOracleEquality}

In this appendix, we prove that the tainted tree oracle generates 
SDTs which are isomorphic to the SDTs generated by the normal tree oracle as defined
in~\cite{Cassel2016}.
In order to do so, we first introduce the constructs used by~\citet{Cassel2016}
for generating SDTs.
We begin with some preliminaries:

%

For a word \(u\) with \(\mathit{Vals}(u) = d_1 \ldots d_k\), we define a \emph{potential} of \(u\).
The potential of \(u\), written as \(\mathit{pot}(u)\),
is the set of indices \(i \in \set{1, \ldots, k}\) for which
there exists no \(j \in \set{1, \ldots, k}\) such that \(j > i\) and \(d_i = d_j\).
The concept of potential essentially allows unique access to a data value,
abstracting away from the concrete position of a data value
in a word.
For a guard $g$ defined over $\mathcal{V}^{+}$ for a word $u$ with $\mathit{Vals}(u) = d_{1}, \ldots d_{k}$,
a representative data value $d^{g}_{u}$
is a data value s.t. $\nu(u) \cup \set{[p \mapsto d_{u}^{g}]} \satisfies g$.
Furthermore, for a word $w = \alpha \cdot w'$ (where $w'$ may be $\epsilon$)
, $w'$ can be represented as $\alpha\inv w$.
The same notation is also extended to sets of words: $\alpha\inv V = \set{\alpha\inv w \mid w \in V}$.


We may now define an SDT\@:
\begin{definition}[Symbolic Decision Tree]\label{def:SDT}
	A Symbolic Decision Tree (SDT) is a register automaton \(\tree = (L,l_0, \mathcal{X}, \Gamma, \lambda)\) where \(L\) and \(\Gamma \) form a tree rooted at \(l_0\).
\end{definition}
For location \(l\) of SDT \(\tree \), we write \(\tree[l]\) to denote the subtree of \(\tree \) rooted at \(l\).
An SDT that results from a tree query \( (u, w) \) (of a prefix word $u$ and a symbolic suffix $w$),
is required to satisfy some canonical form,
captured by the following definition.

\begin{definition}[\((u,w)\)-tree]\label{def:uVTree}
	For any data word $u$ with $k$ actions and any symbolic suffix $w$, a \((u,w)\)-tree is an SDT $\tree$ which has runs
	over all data words in $\dbracket{w}$, and which satisfies the following
	restriction: whenever \( \tuple{l, \alpha(p), g, \pi, l' } \) is the
	\(j^{\mathit{th}}\) transition on some path from $l_0$, then for each
	\( x_i \in \mathcal{X}(l') \) we have either \emph{(i)} \( i < k + j
	\) and \( \pi (x_i) = x_i\), or \emph{(ii)} \( i = k + j \) and \( \pi (x_i) =
	p \).
\end{definition}

If $u = \alpha(d_1) \cdots \alpha_k(d_k)$ is a data word then $\nu_u$ is the valuation of $\set{x_1,\ldots,x_k}$ satisfying $\nu_u(x_i) = d_i$, for $1 \leq i \leq k$.  Using this definition, the notion of a \emph{tree oracle},
which accepts tree queries and returns SDTs, can be described as follows.

\begin{definition}[Tree Oracle]\label{def:treeOracle}
	A tree oracle for a structure $\Structure $ is a function $\oracle$ which,
	for a data language $ \lang $, prefix word $u$ and symbolic
	suffix $w$ returns a \( ( u, w ) \)-tree \( \oracle (\lang, u, w)\) s.t. for any
	word \(  v \in \dbracket{w}\), the following holds: $v$ is accepted by \(
	\oracle(\lang, u, w)\) under $\nu_u$ iff \( u \cdot v \in \lang \).
\end{definition}

A tree oracle returns \emph{equality trees}, defined below:
\begin{definition}[Equality Tree]\label{def:equalityTree}
  An \emph{equality tree} for a tree query $(u, V)$ is a $(u,V)$-tree $\tree$ such that:
  \begin{itemize}
    \item for each action $\alpha$, there is a potential set $I \subseteq
      \mathit{pot}(u)$ of indices such that the initial $\alpha$-guards consist
      of the equalities of form $p = x_i$ for $i \in I$ and one disequality of
      form $\wedge_{i \in I } p \neq x_i$, and
      \item for each initial transition $\tuple{l_0, \alpha(p), g, l}$ of
        $\tree$, the tree $\tree[l]$ is an equality tree for $(u \alpha(d^g_u),
        \alpha\inv V)$.
  \end{itemize}
\end{definition}

\citet{Cassel2016} require their (equality trees) SDTs to be \emph{minimal} (called \emph{maximally abstract} in~\cite{Cassel2016}), i.e., the
SDTs must not contain any redundancies (such as~\Cref{fig:nonMinimalSDT}).
This can be achieved by checking if two sub-trees are equal under some relabelling, and the process of constructing a tree by relabelling an equality sub-tree is called
\emph{specialisation of equality tree}:
\begin{definition}[Specialisation of equality tree]\label{def:treeOracleSpecialisation}
    Let \( \tree \) be an equality tree for prefix \(u\) and set of symbolic suffixes \(V\), and let \(J \subseteq pot(u)\) be a set of indices.
    Then \(\tree \tuple{J}\) denotes the equality tree for \((u,V)\) obtained from \(\tree \) by performing the following transformations for each \(\alpha \):

    \begin{itemize}
        \item Whenever \(\tree \) has several initial \(\alpha \)-transitions of
          form \(\tuple{ l_0, \alpha(p), (p=x_j), l_j }\) with \(j \in J\), then
          all subtrees of form \((\tree [l_j]) \tuple{J [(k+1) \mapsto j]}\) for \(j \in J\) must be defined and isomorphic, otherwise \(\tree \tuple{J}\) is undefined.
              If all such subtrees are defined and isomorphic, then \(\tree \tuple{J}\) is obtained from \(\tree \) by

              \begin{enumerate}
                  \item replacing all initial \(\alpha \)-transitions of form \(\tuple{l_0, \alpha(p), (p=x_j), l_j}\) for \(j \in J\) by the single transition \(\tuple{l_0, \alpha(p), (p=x_m),l_m}\) where \(m = \max(J)\),

                  \item replacing \(\tree [l_m]\) by \((\tree [l_m])
                    \tuple{J[(k+1) \mapsto m]}\), and

                  \item replacing all other subtrees \(\tree [l']\) reached by initial \(\alpha \)-transitions (which have not been replaced in  Step \(1\) by \((\tree [l']) \tuple{J}\).
              \end{enumerate}
    \end{itemize}
\end{definition}

If, for some \(\alpha \), any of the subtrees generated in Step \(2\) or \(3\) are undefined, then \(\tree \tuple{J}\) is also undefined, otherwise \(\tree \tuple{J}\) is obtained after performing Steps \(1-3\) for each \(\alpha \).

\begin{definition}[Necessary Potential set for Tree Oracle]\label{def:necessaryPotSet}
    A necessary potential set \(I\) for the root location \(l_0\) of an equality
    tree \(\oracle(\lang, u,V)\) is a subset of \(\mathit{pot}(u)\) such that for each index \(i \in I\) the following holds:
    \begin{enumerate}
        \item \(\oracle(\lang, u\alpha(d^0_u), V_\alpha) \tuple{\set{i, k+1}}\) is undefined, or
        \item \(\oracle(\lang, u\alpha(d_u^0), V_\alpha) \tuple{\set{i, k+1}}
          \not\iso\oracle(\lang, u\alpha(d_i), V_\alpha)\).
    \end{enumerate}
\end{definition}

Intuitively, a necessary potential set contains indices of data values which influence future behaviour of the SUT\@.
Consequently, indices of data values which do not influence the behaviour of the SUT are excluded from the necessary potential set.
We are now ready to define the tree oracle for equality:
%

\begin{definition}[Tree oracle for equality]\label{def:treeOracleEquality}
    For a language \(\lang\), a prefix \(u\), and the set of symbolic suffixes
    \(V\), the equality tree \(\mathcal{O}(\lang, u,V)\) is constructed as follows:
    \begin{itemize}
        \item If \(V = \set{\epsilon}\), then \(\oracle(\lang, u,\set{\epsilon})\) is the trivial tree with one location \(l_0\) and no registers.
              It is accepting if the word is accpeted, i.e., \(\lambda(l_0) = +\) if \(u \in \lang \), else \(\lambda(l_0) = -\).
              To determine \(u \in \lang \), the tree oracle performs a membership query on \(u\).

        \item If \( V \neq \set{\epsilon}\), then for each \( \alpha \) such that \(V_\alpha = \alpha\inv V\) is non-empty,

              \begin{itemize}
                  \item let \(I\) be the necessary potential set (\Cref{def:necessaryPotSet}),

                  \item \( \oracle(\lang, u,V)\) is constructed as
                    \(\oracle(\lang, u,V) = (L, l_0, \Gamma, \lambda) \), where,
                    letting\\ \(\oracle(\lang, u\alpha(d_i), V_\alpha) \) be the tuple \( (L_i^\alpha, l_{0i}^\alpha, \Gamma_{i}^\alpha, \lambda_i^\alpha) \) for \( i \in (I \cup \set{0}) \),

                        \begin{itemize}
                            \item \(L\) is the disjoint union of all \(L_i^\alpha \) plus an additional initial location \( l_0\),

                            \item \( \Gamma \) is the union of all \( \Gamma_i^\alpha \) for \( i \in (I \cup \set{0})\), and in addition the transitions of form \(\tuple{l_0, \alpha(p), g_i, l_{0i}^\alpha}\) with \(i \in (I \cup \set{0})\), where \(g_i\) is \( \bigwedge_{j \in I} p \neq x_j \) for \( i = 0 \), and \(g_i\) is \(p = x_i\) for \(i \neq 0\), and

                            \item \(\lambda \) agrees with each \(\lambda_i^\alpha \) on \(L_i^\alpha \).
                                  Moreover, if \(\epsilon \in V\), then \(\lambda(l_0) = + \) if \(u \in \lang \), otherwise \(\lambda(l_0)=-\).
                                  Again, to determine whether \(u \in \lang \), the tree oracle performs a membership query for \(u\).
                        \end{itemize}
              \end{itemize}
    \end{itemize}
\end{definition}

\noindent
Intuitively, \(\oracle(\lang, u,V)\) is constructed by joining the trees
\(\oracle(\lang, u\alpha(d_i), V_\alpha)\) with guard \(p = x_i\) for \(i \in
I\), and the tree \(\oracle(\lang, u\alpha(d_u^0), V_\alpha)\) with guard \(\bigwedge_{i \in I} p \neq x_i\), as children of a new root.
Note, while $V$ is a set of symbolic suffixes, RALib technically
handles tree queries sequentially, i.e., as sequential tree queries of prefix $u$ and symbolic suffix $w$.
Consequently, we treat the set of symbolic suffixes $V$ as a singleton,
referred to as `$w$'.

\(\oracle(\lang, u,w)\) is constructed bottom-up, recursively building new `roots' at the top with larger and larger symbolic suffixes (and consequently, shorter and shorter prefixes).
The choice of the necessary potential set \(I\) plays a crucial role: if \(I\)
is larger than necessary, \(\oracle(\lang, u,w)\) contains redundant guards (and is
hence a `non-minimal' SDT).

We now have a clear goal for our proof: we must show that the SDT returned
by~\Cref{alg:minimalSDT} is isomorphic to the SDT returned by the tree oracle for equality as defined
in~\Cref{def:treeOracleEquality} (under the assumption that the `set' of symbolic
suffixes $V$ is a singleton).
We can divide our proof into the following steps:
\begin{enumerate}
  \item We show that~\Cref{alg:decisionQuery} produces a characteristic predicate
    for tree query $(u,w)$, and contains all the information for constructing an
    equality tree,
  \item Next, we show that~\Cref{alg:fullSDT} guarantees that for  potential
    set $I_t$ of a location $l_t$ of the tainted equality tree $\tree_t$, the
    potential set $I$ of equivalent location $l$ of the normal equality tree
    $\tree$ is a subset of $I_{t}$: $I \subseteq I_t$, and finally,
  \item We can then reduce the make the tainted potential set equal to the
    normal potential set (using~\Cref{alg:minimalSDT}) and the resulting tainted equality tree
    will be isomorphic to the normal equality tree.
\end{enumerate}

Each of the above steps correspond to one of our algorithms.
We now begin with step 1: from~\Cref{alg:decisionQuery}, we can state the
following lemmas:
\begin{lemma}[Characteristic Predicate]\label{lemma:characteristicPredicate}
For a tree query $(u,w)$, \Cref{alg:decisionQuery} always produces a
characteristic predicate $H$.
\end{lemma}
\begin{proof}
  We recall that, under the test hypothesis, an SUT $\M$ is deterministic and
  has a finite number of logically disjoint branches to be followed from
  each state.
  \Cref{alg:decisionQuery} initialises two variables $G := \top$ and $H := \bot$.
  For each word $z = u \cdot w$ under a valuation $\nu \satisfies G$, we may
  perform a membership query on $\M$.
  Each query returns the guard $I = \wedge_{i=k+1}^{k+n} \constraints_{\M}(z)[i]$
  such that $\nu \satisfies I$ and the acceptance of the word $z$ in the language
  of $\M$, i.e., $z \in \M$.

  For each iteration of the do-while loop, the variable $G$ is updated with the
  negation of the previously-satisfied guard $I$, i.e., $G := G \wedge \neg I$.
  This guarantees that any new valuation $\nu'$ will not satisfy $I$, and hence,
  the next iteration of the do-while loop shall induce a different run of $\M$.
  Given that $\M$ only has a finite number of logical branches,
  \Cref{alg:decisionQuery} terminates.

  We also know that for each tainted word $z$, we obtain the acceptance of $z \in L(\M)$.
  If $z \in L(\M)$, the variable $H$ is updated to $H \vee I$.
  Therefore, the predicate $H$ returned by \Cref{alg:decisionQuery} is the
  characteristic predicate for the tree query $(u,w)$.
  \qed{}
\end{proof}

After constructing the characteristic predicate, we convert it to a non-minimal SDT
using \Cref{alg:fullSDT}, providing us with the following lemma:
\begin{lemma}[Non-minimal SDT]\label{lemma:nonMinimalSDT}
    For any location \(l_t \) of a non-minimal SDT with an equivalent location
    \(l \) of a minimal SDT, the necessary potential set \(I_t \) of the
    non-minimal SDT is a superset of the necessary potential set \(I \) of the
    minimal SDT\@: \(I \subseteq I_t \subseteq \mathit{pot}(u) \) where
    \(\mathit{pot}(u) \) is the potential of the prefix \(u\) of locations $l_t$
    and $l$.
\end{lemma}
\begin{proof}
    We know that \(I \subseteq \mathit{pot}(u) \) by definition of the necessary potential set.
    For any word $w = u \cdot v$ where the prefix $u$ leads to location $l_t$ of
    the tainted non-minimal SDT,
    \Cref{alg:fullSDT} guarantees that the
    suffixes of $u$ will be classified correctly.
    If the suffixes are classified correctly, we derive that $I_t \supseteq I$
    (otherwise the suffixes will not be classified correctly).
    Since \(I_t \supseteq I\) and \(I, I_t \subseteq \mathit{pot} (u)\), we
    conclude \(I \subseteq I_t \subseteq \mathit{pot} (u)\).
    \qed{}
\end{proof}

Following~\Cref{lemma:nonMinimalSDT}, if we wish to make \(I = I_t \), we can
simply remove all elements from \(I_t \) which do not satisfy the conditions
outlined in~\Cref{def:necessaryPotSet}.
Since we already know that \(I \subseteq I_t \), we can confirm that after
removal of all irrelevant parameters, \(I = I_t \).
\Cref{alg:minimalSDT} accomplishes the same.

\citet{Cassel2016} use the concept of representative data values for
constructing the SDT, while we treat the values symbolically: a representative
data value `represents' the set of data values that satisfy a guard during
construction of the SDT; in our case, we simply let \texttt{Z3} decide on all
the values to use for our membership queries and obtain the guards about them
using their taint markers as identifiers.

\begin{theorem}[Isomorphism of tree oracles]\label{theorem:isoTreeOracles}
  The SDTs generated by the tainted tree oracle and the untainted tree oracle
  for a tree query $(u,w)$ are isomorphic.
\end{theorem}
\begin{proof}
  \Cref{lemma:characteristicPredicate} guarantees that \Cref{alg:decisionQuery}
  returns a characteristic predicate $H$ for the tree query $(u,w)$.
  Application of \Cref{alg:fullSDT} on $H$ constructs a non-minimal SDT\@.
  Using~\Cref{lemma:nonMinimalSDT} and~\Cref{alg:minimalSDT} on the
  non-minimal SDT, we can conclude that the root locations of the tainted tree
  oracle and normal tree oracle have the same necessary potential set.
  By inductive reasoning on the depth of the trees, the same holds for all
  sub-trees of both oracles, eventually reducing to the leaves,
  showing that the tainted tree oracle is isomorphic to tree oracle.
  \qed{}
\end{proof}


\clearpage
\section{Detailed Benchmark results}\label{appendix:benchmarkTable}

\Cref{tbl:benchmarks} contains the full results of the values used to create the
plots from~\Cref{fig:benchmarkPlots}.

\begin{small}
\begin{longtable}{|l|r|r|rrrr|}
\caption{Benchmarks}\label{tbl:benchmarks}\\
\toprule
      Model & Tree Oracle & EQ Oracle & Learn Symbols & Test Symbols & Total Symbols & Learned \\
      & & & (Std. Dev) & (Std. Dev) & (Std. Dev) &  \\
\midrule
\endhead
\midrule
\multicolumn{7}{r}{{Continued on next page}} \\
\midrule
\endfoot

\bottomrule
\endlastfoot
 Abp Output &     Tainted &    Normal &      6.55E+02 &     1.57E+05 &      1.58E+05 &   30/30 \\
            &             &           &    (8.33E+01) &   (1.29E+05) &    (1.29E+05) &         \\
 Abp Output &     Tainted &   Tainted &      6.17E+02 &     1.68E+04 &      1.74E+04 &   30/30 \\
            &             &           &    (7.78E+01) &   (1.15E+04) &    (1.15E+04) &         \\
 Abp Output &      Normal &    Normal &      6.93E+03 &     1.57E+05 &      1.64E+05 &   30/30 \\
            &             &           &    (5.20E+03) &   (1.29E+05) &    (1.29E+05) &         \\
 Abp Output &      Normal &   Tainted &      6.51E+03 &     1.68E+04 &      2.33E+04 &   30/30 \\
            &             &           &    (3.97E+03) &   (1.15E+04) &    (1.29E+04) &         \\
\midrule
     Lock 2 &     Tainted &    Normal &           N-A &          N-A &           N-A &    0/30 \\
            &             &           &         (N-A) &        (N-A) &         (N-A) &         \\
     Lock 2 &     Tainted &   Tainted &      7.10E+01 &     1.15E+03 &      1.22E+03 &   30/30 \\
            &             &           &    (0.00E+00) &   (6.76E+02) &    (6.76E+02) &         \\
     Lock 2 &      Normal &    Normal &           N-A &          N-A &           N-A &    0/30 \\
            &             &           &         (N-A) &        (N-A) &         (N-A) &         \\
     Lock 2 &      Normal &   Tainted &      2.00E+02 &     1.15E+03 &      1.35E+03 &   30/30 \\
            &             &           &    (0.00E+00) &   (6.76E+02) &    (6.76E+02) &         \\
\midrule
     Lock 4 &     Tainted &    Normal &           N-A &          N-A &           N-A &    0/30 \\
            &             &           &         (N-A) &        (N-A) &         (N-A) &         \\
     Lock 4 &     Tainted &   Tainted &      2.41E+02 &     6.29E+03 &      6.53E+03 &   30/30 \\
            &             &           &    (0.00E+00) &   (5.52E+03) &    (5.52E+03) &         \\
     Lock 4 &      Normal &    Normal &           N-A &          N-A &           N-A &    0/30 \\
            &             &           &         (N-A) &        (N-A) &         (N-A) &         \\
     Lock 4 &      Normal &   Tainted &      3.45E+04 &     6.29E+03 &      4.08E+04 &   30/30 \\
            &             &           &    (0.00E+00) &   (5.52E+03) &    (5.52E+03) &         \\
\midrule
     Lock 5 &     Tainted &    Normal &           N-A &          N-A &           N-A &    0/30 \\
            &             &           &         (N-A) &        (N-A) &         (N-A) &         \\
     Lock 5 &     Tainted &   Tainted &      3.80E+02 &     2.62E+04 &      2.66E+04 &   30/30 \\
            &             &           &    (0.00E+00) &   (1.45E+04) &    (1.45E+04) &         \\
     Lock 5 &      Normal &    Normal &           N-A &          N-A &           N-A &    0/30 \\
            &             &           &         (N-A) &        (N-A) &         (N-A) &         \\
     Lock 5 &      Normal &   Tainted &      6.35E+05 &     2.62E+04 &      6.61E+05 &   30/30 \\
            &             &           &    (0.00E+00) &   (1.45E+04) &    (1.45E+04) &         \\
            \\
\midrule
    Fifo 01 &     Tainted &    Normal &      2.90E+01 &     1.71E+01 &      4.62E+01 &   30/30 \\
            &             &           &    (4.08E+00) &   (6.12E+00) &    (6.73E+00) &         \\
    Fifo 01 &     Tainted &   Tainted &      2.97E+01 &     1.38E+01 &      4.35E+01 &   30/30 \\
            &             &           &    (3.83E+00) &   (3.58E+00) &    (4.93E+00) &         \\
    Fifo 01 &      Normal &    Normal &      6.65E+01 &     1.71E+01 &      8.37E+01 &   30/30 \\
            &             &           &    (1.84E+01) &   (6.12E+00) &    (1.80E+01) &         \\
    Fifo 01 &      Normal &   Tainted &      7.07E+01 &     1.38E+01 &      8.46E+01 &   30/30 \\
            &             &           &    (1.74E+01) &   (3.58E+00) &    (1.68E+01) &         \\
\midrule
    Fifo 02 &     Tainted &    Normal &      1.16E+02 &     6.47E+01 &      1.81E+02 &   30/30 \\
            &             &           &    (3.26E+01) &   (2.77E+01) &    (4.28E+01) &         \\
    Fifo 02 &     Tainted &   Tainted &      1.01E+02 &     5.10E+01 &      1.52E+02 &   30/30 \\
            &             &           &    (3.03E+01) &   (1.55E+01) &    (3.31E+01) &         \\
    Fifo 02 &      Normal &    Normal &      3.62E+02 &     6.47E+01 &      4.27E+02 &   30/30 \\
            &             &           &    (1.29E+02) &   (2.77E+01) &    (1.33E+02) &         \\
    Fifo 02 &      Normal &   Tainted &      3.50E+02 &     5.10E+01 &      4.01E+02 &   30/30 \\
            &             &           &    (1.48E+02) &   (1.55E+01) &    (1.49E+02) &         \\
\midrule
    Fifo 03 &     Tainted &    Normal &      3.03E+02 &     1.34E+02 &      4.38E+02 &   30/30 \\
            &             &           &    (8.53E+01) &   (5.84E+01) &    (9.39E+01) &         \\
    Fifo 03 &     Tainted &   Tainted &      2.93E+02 &     1.05E+02 &      3.98E+02 &   30/30 \\
            &             &           &    (8.54E+01) &   (4.69E+01) &    (8.07E+01) &         \\
    Fifo 03 &      Normal &    Normal &      1.64E+03 &     1.34E+02 &      1.78E+03 &   30/30 \\
            &             &           &    (9.00E+02) &   (5.84E+01) &    (8.82E+02) &         \\
    Fifo 03 &      Normal &   Tainted &      1.93E+03 &     1.05E+02 &      2.03E+03 &   30/30 \\
            &             &           &    (1.34E+03) &   (4.69E+01) &    (1.31E+03) &         \\
\midrule
    Fifo 04 &     Tainted &    Normal &      6.87E+02 &     2.20E+02 &      9.06E+02 &   30/30 \\
            &             &           &    (1.51E+02) &   (1.11E+02) &    (2.14E+02) &         \\
    Fifo 04 &     Tainted &   Tainted &      6.35E+02 &     1.62E+02 &      7.96E+02 &   30/30 \\
            &             &           &    (1.41E+02) &   (7.53E+01) &    (1.53E+02) &         \\
    Fifo 04 &      Normal &    Normal &      1.22E+04 &     2.20E+02 &      1.24E+04 &   30/30 \\
            &             &           &    (1.22E+04) &   (1.11E+02) &    (1.22E+04) &         \\
    Fifo 04 &      Normal &   Tainted &      1.19E+04 &     1.62E+02 &      1.20E+04 &   30/30 \\
            &             &           &    (1.21E+04) &   (7.53E+01) &    (1.21E+04) &         \\
\midrule
    Fifo 05 &     Tainted &    Normal &      1.23E+03 &     3.53E+02 &      1.58E+03 &   30/30 \\
            &             &           &    (3.35E+02) &   (2.13E+02) &    (4.49E+02) &         \\
    Fifo 05 &     Tainted &   Tainted &      1.32E+03 &     2.24E+02 &      1.54E+03 &   29/30 \\
            &             &           &    (2.88E+02) &   (9.79E+01) &    (3.14E+02) &         \\
    Fifo 05 &      Normal &    Normal &      1.00E+05 &     3.19E+02 &      1.01E+05 &   25/30 \\
            &             &           &    (1.84E+05) &   (1.67E+02) &    (1.84E+05) &         \\
    Fifo 05 &      Normal &   Tainted &      1.28E+05 &     2.35E+02 &      1.28E+05 &   25/30 \\
            &             &           &    (2.08E+05) &   (8.76E+01) &    (2.08E+05) &         \\
\midrule
 Repetition &     Tainted &    Normal &           N-A &          N-A &           N-A &    0/30 \\
            &             &           &         (N-A) &        (N-A) &         (N-A) &         \\
 Repetition &     Tainted &   Tainted &      1.22E+02 &     7.33E+03 &      7.45E+03 &   30/30 \\
            &             &           &    (0.00E+00) &   (2.03E+03) &    (2.03E+03) &         \\
 Repetition &      Normal &    Normal &           N-A &          N-A &           N-A &    0/30 \\
            &             &           &         (N-A) &        (N-A) &         (N-A) &         \\
 Repetition &      Normal &   Tainted &      8.90E+03 &     7.33E+03 &      1.62E+04 &   30/30 \\
            &             &           &    (1.99E+03) &   (2.03E+03) &    (2.26E+03) &         \\
\midrule
     Set 01 &     Tainted &    Normal &      1.45E+02 &     1.28E+03 &      1.43E+03 &   29/30 \\
            &             &           &    (1.03E+02) &   (1.52E+03) &    (1.52E+03) &         \\
     Set 01 &     Tainted &   Tainted &      9.75E+01 &     1.83E+02 &      2.80E+02 &   30/30 \\
            &             &           &    (3.56E+01) &   (1.61E+02) &    (1.56E+02) &         \\
     Set 01 &      Normal &    Normal &      5.00E+06 &     1.28E+03 &      5.01E+06 &   29/30 \\
            &             &           &    (1.73E+07) &   (1.52E+03) &    (1.73E+07) &         \\
     Set 01 &      Normal &   Tainted &      2.96E+03 &     1.83E+02 &      3.15E+03 &   30/30 \\
            &             &           &    (6.71E+03) &   (1.61E+02) &    (6.69E+03) &         \\
\midrule
     Set 02 &     Tainted &    Normal &      1.61E+03 &     8.21E+03 &      9.82E+03 &   28/30 \\
            &             &           &    (9.96E+02) &   (1.26E+04) &    (1.24E+04) &         \\
     Set 02 &     Tainted &   Tainted &      1.00E+03 &     2.21E+02 &      1.23E+03 &   29/30 \\
            &             &           &    (3.26E+02) &   (2.14E+02) &    (3.68E+02) &         \\
     Set 02 &      Normal &    Normal &      4.61E+06 &     8.60E+03 &      4.62E+06 &   25/30 \\
            &             &           &    (1.43E+07) &   (1.31E+04) &    (1.43E+07) &         \\
     Set 02 &      Normal &   Tainted &      4.35E+04 &     2.20E+02 &      4.37E+04 &   30/30 \\
            &             &           &    (7.28E+04) &   (2.10E+02) &    (7.29E+04) &         \\
\midrule
     Set 03 &     Tainted &    Normal &      1.76E+04 &     5.01E+03 &      2.26E+04 &   24/30 \\
            &             &           &    (8.71E+03) &   (9.51E+03) &    (1.40E+04) &         \\
     Set 03 &     Tainted &   Tainted &      1.44E+04 &     6.91E+02 &      1.51E+04 &   30/30 \\
            &             &           &    (5.05E+03) &   (8.76E+02) &    (4.95E+03) &         \\
     Set 03 &      Normal &    Normal &      5.76E+06 &     3.94E+03 &      5.76E+06 &   14/30 \\
            &             &           &    (1.47E+07) &   (6.48E+03) &    (1.47E+07) &         \\
     Set 03 &      Normal &   Tainted &      2.01E+06 &     2.23E+02 &      2.01E+06 &   28/30 \\
            &             &           &    (3.60E+06) &   (2.06E+02) &    (3.60E+06) &         \\
\midrule
   Sip 2015 &     Tainted &    Normal &      2.14E+03 &     1.89E+05 &      1.92E+05 &   10/30 \\
            &             &           &    (4.00E+02) &   (2.60E+05) &    (2.60E+05) &         \\
   Sip 2015 &     Tainted &   Tainted &      2.30E+03 &     3.18E+04 &      3.41E+04 &   29/30 \\
            &             &           &    (3.13E+02) &   (1.59E+04) &    (1.59E+04) &         \\
   Sip 2015 &      Normal &    Normal &      1.57E+05 &     2.07E+05 &      3.65E+05 &    9/30 \\
            &             &           &    (4.42E+05) &   (2.69E+05) &    (4.81E+05) &         \\
   Sip 2015 &      Normal &   Tainted &      1.47E+05 &     3.18E+04 &      1.79E+05 &   29/30 \\
            &             &           &    (2.80E+05) &   (1.59E+04) &    (2.78E+05) &         \\
\bottomrule
\end{longtable}


\end{small}


\end{subappendices}
\fi

\end{document}